%% file: main.tex
\documentclass[11pt]{article}
\usepackage{graphicx, tikz, multicol, epsfig}
\usepackage{amsthm, amsmath, amssymb, mathtools}
\usepackage{bbm,txfonts}
\usepackage{xifthen}
\usepackage{ wasysym }

\usepackage{array}
\newcolumntype{M}[1]{>{\centering\arraybackslash}m{#1}}

\usepackage[margin=1.2in]{geometry} 
\usepackage[shortlabels]{enumitem}
\usepackage{todonotes}

\usetikzlibrary{calc,shapes, backgrounds, arrows,positioning}

\newtheorem{lemma}{Lemma}
\newtheorem{theorem}{Theorem}
\newtheorem{corollary}{Corollary}
\newtheorem{definition}{Definition}

\newtheorem{remark}{Remark}
\newtheorem{example}{Example}

\newtheorem{proposition}{Proposition}

\newcommand{\ds}{\displaystyle}

\newcommand{\bs}{\backslash}

\newcommand{\lp}{\left (}
\newcommand{\rp}{\right )}

\newcommand{\cC}{\mathcal{C}}
\newcommand{\cD}{\mathcal{D}}
\newcommand{\cE}{\mathcal{E}}
\newcommand{\cF}{\mathcal{F}}
\newcommand{\cH}{\mathcal{H}}

\newcommand{\cM}{\mathcal{M}}

\newcommand{\cX}{\mathcal{X}}

\newcommand{\hgc}{\textrm{(HG)}}
\newcommand{\hlc}{\textrm{(HL)}}
\newcommand{\hpc}{\textrm{(HP)}}
\newcommand{\hfc}{\textrm{(HF)}}

\newcommand{\ugc}{\textrm{(UG)}}
\newcommand{\ulc}{\textrm{(UL)}}
\newcommand{\upc}{\textrm{(UP)}}
\newcommand{\ufc}{\textrm{(UF)}}

\newcommand{\cgc}{\textrm{(CG)}}
\newcommand{\clc}{\textrm{(CL)}}
\newcommand{\cpc}{\textrm{(CP)}}
\newcommand{\cfc}{\textrm{(CF)}}

\newcommand{\Sint}{\textrm{(S5)}}

\newcommand{\ver}{\;\vert\;}

\newcommand{\hh}{\hat{\cH}}
\newcommand{\hg}{\hat{G}}

\begin{document}

\title{On a hypergraph probabilistic graphical model}

\author{
Mohammad Ali Javidian \thanks{University of South Carolina, Columbia, SC 29208,({\tt javidian@email.sc.edu}).}\and	
Linyuan Lu
\thanks{University of South Carolina, Columbia, SC 29208,
({\tt lu@math.sc.edu}). This author was supported in part by NSF
grant DMS-1600811 and ONR grant N00014-17-1-2842.}\and
Marco Valtorta \thanks{University of South Carolina, Columbia, SC 29208,
  ({\tt mgv@cse.sc.edu}). This author was supported in part by ONR grant N00014-17-1-2842.}
\and
Zhiyu Wang \thanks{University of South Carolina, Columbia, SC 29208,
({\tt zhiyuw@email.sc.edu}).} 
}

\maketitle

\begin{abstract}
    
We propose a directed acyclic hypergraph framework for a probabilistic graphical model that we call \textit{Bayesian hypergraphs}. The space of directed acyclic hypergraphs is much larger than the space of chain graphs. Hence Bayesian hypergraphs can model much finer factorizations than Bayesian networks or LWF chain graphs and provide simpler and more computationally efficient procedures for factorizations and interventions.
Bayesian hypergraphs also allow a modeler to represent causal patterns of interaction such as Noisy-OR graphically (without additional annotations). We introduce global, local and pairwise Markov properties of Bayesian hypergraphs and prove under which conditions they are equivalent.  We define a projection operator, called shadow, that maps Bayesian hypergraphs to chain graphs,  and show that the Markov properties of a Bayesian hypergraph are equivalent to those of its corresponding chain graph. We extend the causal interpretation of LWF chain graphs to Bayesian hypergraphs and provide corresponding formulas and a graphical criterion for intervention.
\end{abstract}

\section{Introduction}

Probabilistic graphical models are graphs in which nodes represent random variables and edges represent conditional independence assumptions. They provide a compact way to represent the joint probability distributions of a set of random variables. In \textit{undirected} graphical models, e.g., Markov networks (see \cite{DLS, Pearl88}), there is a simple rule for determining independence: two set of nodes $A$ and $B$ are conditionally independent given $C$ if removing $C$ separates $A$ and $B$. On the other hand, \textit{directed} graphical models, e.g. Bayesian networks (see \cite{KSC, WL, Pearl88}), which consist of a directed acyclic graph (DAG) and a corresponding set of conditional probability tables, have a more complicated rule (d-separation) for determining independence. 
More complex graphical models include various types of graphs with edges of several types (e.g., \cite{CW, WCP, RS, Pena}), including chain graphs \cite{LW, Lau}, for which different interpretations have emerged \cite{AMP, Drton}. 


	 

Probabilistic Graphical Models (PGMs) enjoy a well-deserved popularity because they allow explicit representation of structural constraints in the language of graphs and similar structures.  From the perspective of efficient belief update, factorization of the joint probability distribution of random variables corresponding to variables in the graph is paramount, because it allows decomposition of the calculation of the evidence or of the posterior probability \cite{LJ}.  The proliferation of different PGMs that allow factorizations of different kinds leads us to consider a more general graphical structure in this paper, namely directed acyclic hypergraphs.  Since there are many more hypergraphs than DAGs, undirected graphs, chain graphs, and, indeed, other graph-based networks, as discussed in Remark \ref{counting}, Bayesian hypergraphs can model much finer factorizations and thus are more computationally efficient.  When tied to probability distributions, directed acyclic hypergraphs specify independence (and possibly other) constraints through their Markov properties; we call the new PGM resulting from the directed acyclic hypergraphs and their Markov properties \textit{Bayesian hypergraphs}.  We provide such properties and show that they are consistent with the ones used in Bayesian networks, Markov networks, and LWF chain graphs, when the directed acyclic hypergraphs are suitably restricted. We show in Section \ref{sec:BH-fact} that there are situations that may be of interest to a probabilistic or causal modeler that can be modeled more explicitly using Bayesian hypergraphs; in particular, some causal patterns, such as independence of causal influence (e.g., Noisy-OR), can be expressed graphically in Bayesian hypergraphs, while they require a numerical specification in DAGs or chain graphs.
We provide a causal interpretation of Bayesian hypergraphs that extends the causal interpretation of LWF chain graphs \cite{LR}, by giving  corresponding formulas and a graphical criterion for intervention.

The paper is organized as follows: In Section \ref{sec:defns}, we  introduce some common notations, terminology and concepts on graphs and hypergraphs. In Section \ref{sec:UD} and Section \ref{sec:CG}, we review the Markov properties and factorizations in the case of undirected graphs and chain graphs. In Section \ref{sec:HBN}, we  introduce the Bayesian hypergraphs model, discuss the factorizations, Markov properties and its relations to chain graphs. In Section \ref{sec:intervention}, we discuss how interventions can be achieved in Bayesian hypergraphs.
Section 7 concludes the paper and includes some directions for further work.

\section{Terminology and concepts}\label{sec:defns}

In this paper, we use $[n]$ to denote the set $\{1,2,\ldots, n\}$. For $a, b\in \mathbb{Z}$, We use $[a,b]$ to denote $\{k\in \mathbb{Z}: a\leq k\leq b\}$.
Given a set $h$, we use $|h|$ to denote the number of elements in $h$.
\subsection{Graphs}

A graph $G=(V,E)$ is an ordered pair $(V,E)$ where $V$ is a finite set of \textit{vertices} (or \textit{nodes}) and $E \subseteq V\times V$ consists of a set of ordered pairs of vertices $(v,w) \in V\times V$. Given a graph $G$, we will use $V(G), E(G)$ to denote the set of vertices and edges of $G$ respectively.
An edge $(v,w) \in E$ is \textit{directed} if $(w,v) \notin E$ and \textit{undirected} if $(w,v) \in E$. We write $v\to w$ if $(v,w)$ is directed and $v - w$ if $(v,w)$ is undirected. 
If $v -w$ then we call $v$ a \textit{neighbor} of $w$ and vice versa. If $v \to w$, then we call  $v$ a \textit{parent} of $w$ and $w$ a \textit{child} of $v$. 
Let $pa_G(v)$ and $nb_G(v)$ denote the set of parents and neighbors of $v$, respectively. 
We say $v$ and $w$ are \textit{adjacent} if either $(v,w) \in E$ or $(w,v)\in E$, i.e., either $v\to w$, $w \to v$ or $v - w$. We say an edge $e$ is \textit{incident} to a vertex $v$ if $v$ is contained in $e$.
We also define the \textit{boundary} $bd(v)$ of $v$ by
    $$bd(v) = nb(v) \cup pa(v).$$ 
Moreover, given $\tau \subseteq V$, define 

$$pa_G(\tau) = \lp \ds\bigcup_{v\in \tau} pa_G(v)\rp \backslash \tau.$$
$$nb_G(\tau) = \lp \ds\bigcup_{v\in \tau} nb_G(v)\rp \backslash \tau.$$
$$bd_G(\tau) = \lp \ds\bigcup_{v\in \tau} bd_G(v)\rp \backslash \tau.$$
$$cl_G(\tau) = bd_G(\tau) \cup \tau.$$

For every graph $G=(V,E)$, we will denote the \textit{underlying undirected graph} $G^u =(V,E^u)$, i.e., $E^u = \{(v,u): (v,u) \in E \textrm{ or } (u,v) \in E\}$. 
A \textit{path} in $G$ is a sequence of distinct vertices $v_0, \ldots, v_k$ such that $(v_{i},v_{i+1}) \in E$ for all $0\leq i\leq k-1$. A path $v_0, \ldots ,v_k$ is \textit{directed} if for all $0\leq i <k$, $(v_i, v_{i+1})$ is a directed edge, i.e., $(v_i, v_{i+1}) \in E$ but $(v_{i+1}, v_{i}) \notin E$. 
A \textit{cycle} is a path with the modification that $v_k = v_0$. A cycle is \textit{partially directed} if at least one of the edges in the cycle is a directed edge.
A graph $G$ is \textit{acyclic} if $G$ contains no  partially directed cycle.
A vertex $v$ is said to be an \textit{anterior} of a vertex $u$ if there is a path from $v$ to $u$. We remark that every vertex is also an anterior of itself.  
If there is a directed path from $v$ to $u$, we call $v$ an \textit{ancestor} of $u$ and $u$ a \textit{descendent} of $v$. Moreover, $u$ is a \textit{non-descendent} of $v$ if $u$ is not a descendent of $v$. 
Let $ant(u)$ and $an(u)$ denote the set of anteriors and ancestors of $u$ in $G$ respectively. Let $de(v)$ and $nd(v)$ denote the set of descendents and non-descendents of $v$ in $G$ respectively.
For a set of vertices $\tau$, we also define $ant(\tau) = \{ant(v): v \in \tau\}.$
Again, note that $\tau \subseteq ant(\tau)$.

A \textit{subgraph} of a graph $G$ is a graph $H$ such that $V(H)\subseteq V(G)$ and each edge present in $H$ is also present in $G$ and has the same type. An \textit{induced subgraph} of $G$ by a subset $A \subseteq V(G)$, denoted by $G_A$ or $G[A]$, is a subgraph of $G$ that contains all and only vertices in $A$ and all edges of $G$ that contain only vertices in $A$. A \textit{clique} or \textit{complete graph} with $n$ vertices, denoted by $K_n$, is a graph such that every pair of vertices is connected by an undirected edge.

Now we can define several basic graph representations used in probabilistic graphical models.
An \textit{undirected graph} is a graph such that every edge is undirected. A \textit{directed acyclic graph (DAG)} is a graph such that every edge is directed and contains no directed cycles. A \textit{chain graph} is a graph without partially directed cycles.
Define two vertices $v$ and $u$ to be equivalent if there is an undirected path from $v$ to $u$. Then the equivalence classes under this equivalence relation are the \textit{chain components} of $G$. 
For a vertex set $S$, define $E^*(S)$ as the edge set of the complete undirected graph on $S$.
Given a graph $G=(V,E)$ with chain components $\{\tau: \tau \in \cD\}$, the \textit{moral graph} of $G$, denoted by $G^m = (V,E^m)$, is a graph such that $V(G^m)=V(G)$ and $E^m = E^u \cup \bigcup_{\tau\in \cD} E^*(bd(\tau))$, i.e., the underlying undirected graph, where the boundary w.r.t. $G$ of every chain component is made complete. The moral graphs are natural generalizations to chain graphs of the similar concept for DAGs given in \cite{LS} and \cite{Letal}. 


\subsection{Hypergraphs}

Hypergraphs are generalizations of graphs such that each edge is allowed to contain more than two vertices. Formally, an \textit{(undirected) hypergraph} is a pair $\cH = (V, \cE)$, where $V= \{v_1, v_2, \cdots, v_n\}$ is the set of \textit{vertices} (or \textit{nodes}) and $\cE = \{h_1, h_2, \cdots, h_m\}$ is the set of \textit{hyperedges} where $h_i \subseteq V$ for all $i \in [m]$. If $|h_i| = k$ for every $i\in [m]$, then we say $\cH$ is a \textit{$k$-uniform} (undirected) hypergraph. 
A \textit{directed hyperedge} or \textit{hyperarc} $h$ is an ordered pair, $h = (X,Y)$, of (possibly empty) subsets of $V$ where $X\cap Y = \emptyset$; $X$ is the called the \textit{tail} of $h$ while $Y$ is the \textit{head} of $h$. We write $X= T(h)$ and $Y = H(h)$. We say a directed hyperedge $h$ is \textit{fully directed} if none of $H(h)$ and $T(h)$ are empty.
A \textit{directed hypergraph} is a hypergraph such that all of the hyperedges are directed. A \textit{$(s,t)$-uniform directed hypergraph} is a directed hypergraph such that the tail and head of every directed edge have size $s$ and $t$ respectively. For example, any DAG is a $(1,1)$-uniform hypergraph (but not vice versa). An undirected graph is a $(0,2)$-uniform hypergraph. Given a hypergraph $\cH$, we use $V(\cH)$ and $E(\cH)$ to denote the the vertex set and edge set of $\cH$ respectively. 

We say two vertices $u$ and $v$ are {\em co-head} (or {\em co-tail}) if there is a directed hyperedge $h$ such that
$\{u,v\}\subset H(h)$ ( or $\{u,v\}\subset T(h)$ respectively). 
Given another vertex $u\neq v$, we say $u$ is a \textit{parent} of $v$,
denoted by $u\to v$, if there is
a directed hyperedge $h$ such that $u\in T(h)$ and $v\in H(h)$. If $u$ and $v$ are co-head, then $u$ is a \textit{neighbor} of $v$. If $u,v$ are neighbors, we denote them by $u -v$. Given $v \in V$, we define parent ($pa(v)$), neighbor ($nb(v)$), boundary ($bd(v)$), ancestor ($an(v)$), anterior ($ant(v)$), descendant ($de(v)$), and non-descendant ($nd(v)$) for hypergraphs exactly the same as for graphs (and therefore use the same names). The same holds for the equivalent concepts for $\tau \subseteq V$. Note that it is possible that some vertex $u$ is both the parent and neighbor of $v$. 

A {\em partially directed cycle} in $\cH$ is a sequence $\{v_1, v_2, \dots v_k\}$ satisfying that $v_i$ is either a neighbor or a parent of $v_{i+1}$ for all $1\leq i \leq k$ and $v_i \to v_{i+1}$ for some $1\leq i \leq k$. Here $v_{k+1} \equiv v_1$. We say a directed hypergraph $\cH$ is {\em acyclic} if $\cH$ contains no partially directed cycle. For ease of reference, we call a directed acyclic hypergraph a \textit{DAH} or a \textit{Bayesian hypergraph structure} (as defined in Section \ref{sec:HBN}). Note that for any two vertices $u,v$ in a directed acyclic hypergraph $\cH$, $u$ can not be both the parent and neighbor of $v$ otherwise we would have a partially directed cycle.

\begin{remark}
DAHs are generalizations of undirected graphs, DAGs and chain graphs. In particular an undirected graph can be viewed as a DAH in which every hyperedge is of the form $(\emptyset, \{u,v\})$. A DAG is a DAH in which every hyperedge is of the form $(\{u\},\{v\})$. A chain graph is a DAH in which every hyperedge is of the form $(\emptyset, \{u,v\})$ or $(\{u\},\{v\})$.
\end{remark}

We define the \textit{chain components} of $\cH$ as the equivalence classes under the equivalence relation where two vertices $v_1,v_t$ are equivalent if there exists a sequence of distinct vertices $v_1, v_2, \ldots, v_t$ such that $v_i$ and $v_{i+1}$ are co-head for all $i\in [t-1]$. The chain components $\{\tau : \tau \in \cD\}$ yields an unique natural partition of the vertex set $V(\cH) = \bigcup_{\tau \in \cD} \tau$ with the following properties:

\begin{proposition}\label{DAG-partition}
	Let $\cH$ be a DAH and $\{\tau : \tau \in \cD\}$ be its chain components. Let $G$ be a graph obtained from $\cH$ by contracting each element of $\{\tau : \tau \in \cD\}$ into a single vertex and creating a directed edge from $\tau_i \in V(G)$ to $\tau_j \in V(G)$ in $G$ if and only if there exists a hyperedge $h \in E(\cH)$ such that $T(h) \cap \tau_i \neq \emptyset$ and $H(h) \cap \tau_j \neq \emptyset$. Then $G$ is a DAG.
\end{proposition}

\begin{proof}

First of all, clearly $G$ is a directed graph. Now since $\cH$ is a DAH, there is no directed hyperedge such that both its head and tail intersect a common chain component. Hence $G$ has no self-loop. It remains to show that there is no directed cycle in $G$. Supporse for contradiction that there is a directed cycle $\tau_1, \tau_2, \ldots, \tau_k$ in $G$. Then by the construction of $G$, there is a sequence of hyperedges $\{h_1, h_2, \ldots, h_k\}$  such that $T(h_i) \cap \tau_i \neq \emptyset$ and $H(h_i) \cap \tau_{i+1} \neq \emptyset$ (with $\tau_{k+1} \equiv \tau_1$). Since there is a path between any two vertices in the same component, it follows that there is a partially directed cycle in $\cH$, which contradicts that $\cH$ is acyclic. Hence we can conclude that $G$ is indeed a DAG.
	
\end{proof}

Note that the DAG obtained in Proposition \ref{DAG-partition} is unique and given a DAH $\cH$ we call such $G$ the \textit{canonical DAG} of $\cH$.
A chain component $\tau$ of $\cH$ is \textit{terminal} if the out degree of $\tau$ in $G$ is $0$, i.e., there is no $\tau'\neq \tau$ such that $\tau \to \tau'$ in $G$. A chain component $\tau$ is \textit{initial} if the in degree of $\tau$ in $G$ is $0$, i.e., there is no $\tau'\neq \tau$ such that $\tau' \to \tau$ in $G$. We call a vertex set $A \subseteq V(\cH)$ an \textit{anterior set} if it can be generated by stepwise removal of terminal chain components.  We call $A$ an \textit{ancestral set} if $bd(A) = \emptyset$ in $\cH$. We remark that given a set $A$, $ant(A)$ is also the smallest ancestral set containing $A$.

A sub-hypergraph of $\cH=(V, \cE)$ is a directed hypergraph $\cH' = (V',\cE')$ such that $V' \subseteq V$ and $\cE' \subseteq \cE$. Given $S\subseteq V(\cH)$, we say a directed hypergraph $\cH'$ is a sub-hypergraph of $\cH$ induced by $S$, denoted by $\cH_S$ or $\cH[S]$, if $V(\cH') = S$ and $h\in E(\cH')$ if and only if $h \in E(\cH)$ and $H(h)\cup T(h) \subseteq S$. 

To illustrate the relationship between a directed acyclic hypergraph and a chain graph, we will introduce the concept of a \textit{shadow} of a directed acyclic hypergraph. Given a directed acyclic hypergraph $\cH$,  let the \textit{(directed) shadow} of $\cH$, denoted by $\partial(\cH)$, be a graph $G$ such that $V(G) = V(\cH)$, and for every hyperedge $h=(X,Y) \in E(\cH)$, $G[Y]$ is a clique (i.e. every two vertices in $G[Y]$ are neighbors) and there is a directed edge from each vertex of $X$ to each vertex of $Y$ in $G$.

\begin{proposition}\label{prop:shadow-lem}
Suppose $\cH$ is a directed acyclic hypergraph and $G$ is the shadow of $\cH$. Then 
\begin{enumerate}[(i)]
    \item $G$ is a chain graph.
    \item For every vertex $v\in V(\cH)=V(G)$, $nb_G(v) = nb_{\cH}(v)$ and $pa_G(v)= pa_{\cH}(v)$.
\end{enumerate}
\begin{proof}
For $(i)$, note that since $\cH$ is acyclic, there is no partially directed cycle in $\cH$. It follows by definition that there is no partially directed cycle in $G$. Hence, the shadow of a directed acyclic hypergraph is a chain graph. $(ii)$ is also clear from the definition of the shadow.
\end{proof}
\end{proposition}

\begin{figure}[htb]
	\begin{center}
		\begin{minipage}{.2\textwidth}
			\resizebox{3cm}{!}{\input{images/HG_ex1.tikz}}
		\end{minipage}
			\hspace{1cm}
		\begin{minipage}{.2\textwidth}
			\resizebox{3cm}{!}{\input{images/HG_ex1_DAG.tikz}}
		\end{minipage}
		\hspace{1cm}
		\begin{minipage}{.2\textwidth}
		\resizebox{3cm}{!}{\input{images/HG_ex1_shadow.tikz}}
		\end{minipage}
		\caption{(1) a DAH $\cH$. (2) the canonical DAG of $\cH$. (3) the shadow of $\cH$.}
		\label{fig:0}
	\end{center} 
\end{figure}
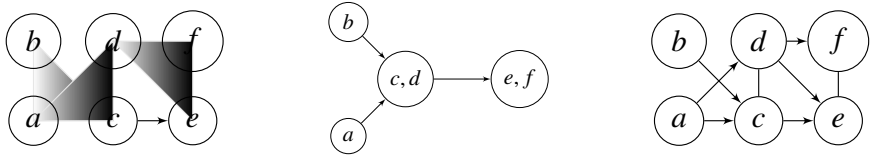

\subsection{Hypergraph drawing}
In this subsection, we present how directed edges are drawn in this paper and illustrate the concepts with an example. 
For a fully directed hyperedge with two vertices (both head and tail contain exactly one vertex), we use the standard arrow notation. For a fully directed hyperedge with at least three vertices, we use a shaded polygon to represent that edge, with the darker side as the head and the lighter side as the tail. For hyperedges of the type $(\emptyset, A)$, we use an undirected line segment (i.e. $-$) to denote the hyperedge if $|A| = 2$ and a shaded polygon with uniform gray color if $|A| \geq 3$.
For example, in Figure \ref{fig:0}, the directed hyperedges are $(\{a,b\},\{c\}),(\{a\},\{c,d\}),(\{d\},\{e,f\}), (\{c\},\{e\})$.  Here $a$ and $b$ are co-tail, $c$ and $d$, $e$ and $f$ are co-head. Figure \ref{fig:0} (2) shows the canonical DAG associated to $\cH$ with four chain components:$\{a\},\{b\},\{c,d\},\{e,f\}$. Figure \ref{fig:0} (3) shows the shadow of $\cH$.

\subsection{Construction of a directed acyclic hypergraph from chain graph}\label{sec:hypergraph-extension}

In this subsection, we show how to construct a directed acyclic hypergraph from a chain graph according to the LWF interpretation. Due to the expressiveness and generality of a directed hypergraph, other constructions may exist too. 
Let $G$ be a chain graph with $n$ vertices. We will explicitly construct a directed acyclic hypergraph $\cH$ on $n$ vertices that correspond to $G$. We remark that the construction essentially creates a hyperedge for each maximal clique in the moral graph of $G_{cl(\tau)}$ for every chain component $\tau$ of $\cH$.\\

\noindent\textit{Construction:}
$$V(\cH) = V(G).$$
The edge set of $\cH$ is constructed in two phases:
\begin{description}
\item Phase I: 
\begin{itemize}
    \item For each $v \in V(G)$, let $S_v$ be the set of children of $v$ in $G$. Consider the subgraph $G'$ of $G[S_v]$ induced by the undirected edges in $G[S_v]$. For each maximal clique (with vertex set $K$) in $G'$, add the directed hyperedge $(\{v\},K)$ into $\cH$.
    \item Let $\cH'$ be the resulting hypergraph after performing the above procedure for every $v\in V(G)$. Now for every maximal clique $K$ (every edge in $K$ is undirected) in $G$, if $K \not\subseteq H(h)$ for every $h\in E(\cH')$, add the directed hyperedge $(\emptyset, K)$ into $\cH$.
\end{itemize}

\item Phase II: Let $\cH'$ be the resulting hypergraph constructed from Phase I and $\{\tau: \tau \in \cD\}$ be the chain components of $G$.
Given $\tau$, let $\cH^*_{\tau}$ be the set of edges $h$ in $\cH'_{cl(\tau)}$ such that $H(h) \cap \tau \neq \emptyset$.

Define 
$$ E(\cH) =\bigcup_{\tau \in \cD}\Bigg\{ 
\Bigg( \ds\bigcup_{\substack{h \in E(\cH^*_{\tau})\\B\subseteq H(h)}} T(h), B \Bigg)
: B= \bigcap_{h \in \cF} H(h), \cF \subseteq E(\cH^*_{\tau}) \Bigg\}.$$

\end{description}

Note that the resulting hypergraph $\cH$ is a directed acyclic hypergraph since a partially directed cycle $C$ in $\cH$ corresponds to a directed cycle in $G$.
Moreover, the above construction gives us an injection from the family of chain graphs with $n$ vertices to the family of directed acyclic hypergraphs with $n$ vertices. 

\begin{figure}[htb]
	\begin{center}
		\begin{minipage}{.2\textwidth}
			\resizebox{3cm}{!}{\input{images/CG_ex2.tikz}}
		\end{minipage}
		\hspace{1cm}
		\begin{minipage}{.2\textwidth}
			\resizebox{3cm}{!}{\input{images/HBN_ex2_moral.tikz}}
		\end{minipage}
		\caption{(1) a simple chain graph $G$;  (2) the corresponding DAH of $G$ in the LWF interpretation.}
		\label{f1}
	\end{center} 
\end{figure}
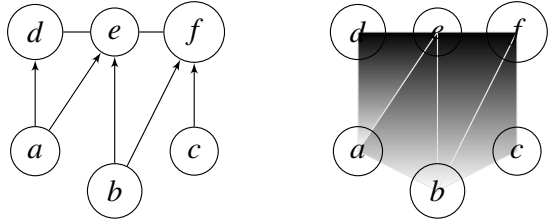

Figure \ref{f1} contains an example of a simple chain graph and its corresponding version in the hypergraph representation. Recall that every fully directed hyperedge is represented (in the drawing) by a colored convex region. The darker side is the head and the lighter side is the tail.
We will detail the hyperedges existing in every phase of the construction:
\begin{itemize}
\item Phase I: the hyperedges in $\cH$ are $\{a,d,e\}$ and $\{b,e,f\}$ and $\{c,f\}$.
\item Phase II: For each chain component $\tau$, we obtain all subsets $B$ of $\tau$ which are the intersections of the heads of the hyperedges intersecting $\tau$. For each such $B$ obtained, create a hyperedge whose head is $B$ and whose tail is the union of  the tails of the hyperedges containing $B$ in its head. In Figure \ref{f1}, the set of such $B$'s are $\{d,e\},\{e\},\{e,f\}$, $\{f\}$. Hence 
$$E(\cH) = \Bigg\{ \{a,d,e\},\{e,a,b\},\{b,e,f\},\{b,c,f\}\Bigg\}.$$
\end{itemize}

Hence the resulting hypergraph $\cH$ is the one in Figure \ref{f1}(2).

For ease of reference, given a chain graph, we will call the hypergraph $\cH$ constructed above the \textit{canonical LWF  DAH} of $G$. We say $\cH$ is \textit{hypermoralized} from $G$ if $\cH$ is the canonical LWF DAH of $G$.
Moreover, we call the family of all such hypergraphs (i.e. the canonical LWF DAH of some chain graph) \textit{LWF DAHs}.

\begin{figure}[htb]
	\begin{center}
    \includegraphics[width=.50\textwidth]{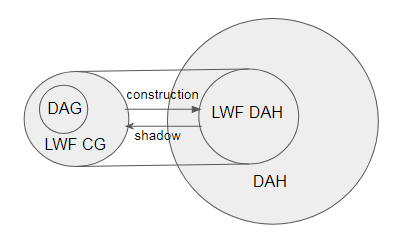}
    \caption{Relationship between chain graphs and directed acyclic hypergraphs}
		\label{f:inc}
    \end{center}
\end{figure}

\begin{remark}
In this section, we gave an injective mapping from the space of chain graphs to the space of directed acyclic hypergraphs such that the LWF DAHs have the same Markov properties as LWF chain graphs. We believe some other types of chain graphs can be modeled by DAHs too (e.g. MVR DAHs) but we do not explore them in this paper.
\end{remark}

We will summarize the relations between a chain graph and its canonical LWF DAH in the following lemma:

\begin{lemma}\label{lem:comp}
Let $G$ be a chain graph and $\cH$ be its canonical LWF DAH. Then we have
\begin{enumerate}[(i)]
    \item\label{comp1} For each vertex $v \in V(G)=V(\cH)$,
             $nb_G(v) = nb_{\cH}(v)$ and 
            $pa_G(v) = pa_{\cH}(v)$.
    \item\label{comp2} $G$ is the shadow of $\cH$.
    \item\label{comp3} $\cH$ is a directed acyclic hypergraph.
\end{enumerate}
\begin{proof}

We will first show \ref{comp1}. Note that by our construction in Phase I, if two vertices are neighbors in $G$, then they are co-head in $\cH$. Moreover, if $u$ is the parent of $v$ in $G$, then $u$ is still the parent of $v$ in $\cH$. These relations remain true in Phase II.
Hence we obtain that $nb_G(v) \subseteq nb_{\cH}(v)$ and $pa_G(v) \subseteq pa_{\cH}(v)$ for all $v\in V(\cH)$. It remains to show that for each $v \in V(\cH)$, no additional neighbor or parent of $v$ (compared to the case in $G$) is added in the construction. 
In Phase I , every hyperedge added is either of the form $(\emptyset, K)$ or $(\{w\},K)$ where $K\subseteq V$ induces a complete undirected graph in $G$ and $w$ is the parent of every element in $K$. Hence for every $v\in V(\cH)$, no additional neighbor or parent of $v$ is added in Phase I. 
Now let us examine Phase II. Given an edge $h = (A,B)\in E(\cH)$, there exists some $\tau \in \cD(G)$ such that $B = \bigcap_{h \in \cF} H(h)$ for some $\cF \subseteq E(\cH^*_{\tau})$. Moreover, $A = \bigcup_{\substack{h \in E(\cH^*_{\tau}))\\B\subseteq H(h)}} T(h)$. Note for every pair of elements $u,v\in B$, $u,v$ are already neighbors in $G$ since $u,v \in H(h)$ for some $h \in \cF$ from Phase I. Moreover, for every $v\in A, u\in B$, $v$ is already a parent of $u$ in $G$ since there exists some $h$ constructed in Phase I such that $u\in H(h)$ and $v\in T(h)$. Therefore, it follows that any edge defined in Phase II does not create any new neighbor or parent for any $v \in V(G)$. Thus, we can conclude that for all $v\in V(G)=V(\cH)$, $nb_G(v) = nb_{\cH}(v)$ and $pa_G(v) = pa_{\cH}(v)$.

\ref{comp2} is implied by \ref{comp1} by the definition of a shadow.
\ref{comp3} is implied by \ref{comp2} since $G$ is acyclic and $G$ is the shadow of $\cH$.
\end{proof}
\end{lemma}

\section{Markov properties for undirected graphs}\label{sec:UD}

In this section, we will summarize some basic results on the Markov properties of undirected graphs. Let us first introduce some notations. In the rest of this week, let $(X_{\alpha})_{\alpha\in V}$ be a collection of random variables taking values in some product space $\cX =  \times_{\alpha\in V} \cX_{\alpha}$. Let $P$ denote a probability measure on $\cX$. For a subset $A$ of $V$, we use $\cX_A$ to denote $\ds\times_{\alpha\in A} \cX_{\alpha}$ and $P_A$ is the marginal measure on $\cX_A$. A typical element of $\cX_A$ is denoted by $x_A = (x_{\alpha})_{\alpha\in A}$. We will use the short notation $A \Perp B \ver C$ for $X_A \Perp X_B \ver X_C$.

Recall that an independence model $\Perp$ is a ternary relation over subsets of a finite set $V$. The following properties have been defined for the conditional independences of probability distributions. Let $A,B,C,D$ be disjoint subsets of $V$ where $C$ may be the empty set.

\begin{description}
\item[S1 (Symmetry)]\label{S1} $A\Perp B \ver C \implies B\Perp A \ver C$;
\item[S2 (Decomposition)]\label{S2} $A \Perp BD \ver C \implies \lp A \Perp B \ver C \textrm{ and } A \Perp D \ver C \rp$;	
\item[S3 (Weak Union)]\label{S3} $A \Perp BD \ver C \implies \lp A\Perp B \ver DC \textrm{ and } A\Perp D \ver BC\rp$;
\item[S4 (Contraction)]\label{S4} $\lp A \perp B \ver DC \textrm{ and } A \Perp D \ver C\rp \iff A \Perp BD \ver C$;
\item[S5 (Intersection)]\label{S5} $\lp A \perp B \ver DC \textrm{ and } A \Perp D \ver BC\rp \implies A \Perp BD \ver C$;
\item[S6 (Composition)]\label{S6} $\lp A \perp B \ver C \textrm{ and } A \Perp D \ver C\rp \iff A \Perp BD \ver C$;
\end{description}

An independence model is a \textit{semi-graphoid} if it satisfies the first four independence properties listed above. A discussion of conditional independence can be found in Dawid \cite{Dawid} where it is shown that any probability measure is a semi-graphoid. Also see Studeny \cite{Studeny92} and Pearl \cite{Pearl88} for a discussion of soundness and (lack of) completeness of these axioms.
If a semi-graphoid further satisfies the intersection property, we say it is a \textit{graphoid}. A \textit{compositional graphoid} further satisfies the composition property.
We follow the same naming convention as Frydenberg \cite{Fry}. Given an undirected graph $G$, we say $C$ separates $A$ and $B$ in $G$ if there is no path from any vertex in $A$ to any vertex in $B$ in $G[V(G)\bs C]$. If $G$ is an undirected graph, then a probability measure $P$ is said to be:
\begin{description}
\item[(UP)] \textit{pairwise $G$-Markovian} if $\alpha \Perp \beta \ver V\bs \{\alpha, \beta\}$ whenever $\alpha$ and $\beta$ are non-adjacent in $G$. 
\item[(UL)] \textit{local $G$-Markovian} if $\alpha \Perp V\bs cl(\alpha) \ver bd(\alpha)$ for all $\alpha \in V(G)$.
\item[(UG)] \textit{global $G$-Markovian} if $A\Perp B \ver C$ whenever $C$ separates $A$ and $B$ in $G$.
\end{description}

The following theorem by Pearl and Paz \cite{Pearl86} gives a sufficient condition for the equivalence of $\ugc$, $\ulc$ and $\upc$.

\begin{theorem}{(\cite{Pearl86})}
If $G$ is an undirected graph and $P$ satisfies (S5), then $\ugc$, $\ulc$ and $\upc$ are equivalent and $P$ is said to be $G$-Markovian if they hold.
\end{theorem}

Conditional independences and thus Markov properties are closely related to factorizations. A probability measure $P$ on $\cX$ is said to \textit{factorize} according to $G$ if for each clique $h$ in $G$, there exist a non-negative function $\psi_h$ depending on $x_h$ only and there exists a product measure $\mu = \times_{\alpha \in V} \;\mu_{\alpha}$ on $\cX$ such that $P$ has density $f$ with respect to $\mu$ where $f$ has the form

\begin{equation}
    f(x) = \ds\prod_{h\in \cC} \psi_h(x)
\end{equation}
where $\cC$ is the set of maximal cliques in $G$. If $P$ factorizes according to $G$, we say $P$ has property (UF).
It is known (see \cite{Lau}) that 

$$\textrm{\ufc $\implies$ \ugc $\implies$ \ulc $\implies$ \upc}.$$

Moreover, in the case that $P$ has a positive and continuous density, it can be proven using M\"{o}bius inversion lemma that $\upc \implies \ufc$. This result seems to have been discovered in various forms by a number of authors \cite{Speed} and is usually attributed to Hammersley and Clifford \cite{HC} who proved the result in the discrete case.



\section{Markov properties of chain graphs}\label{sec:CG}

We use the same notations as Section \ref{sec:UD}.
Let $G$ be a chain graph and $P$ be a probability measure defined on some product space $\cX = \times_{\alpha \in V(G)} \cX_{\alpha}$. Then $P$ is said to be 

\begin{description}

\item[(CP)] \textit{pairwise $G$-Markovian}, if for every pair $(v,u)$ of non-adjacent vertices with $u \in nd(v)$, 
	\begin{equation}
	v \Perp u \ver nd(v) \backslash \{v,u\}.
	\end{equation}

\item[(CL)] \textit{local $G$-Markovian}, relative to $G$, if for any vertex $v \in V(G)$, 
	\begin{equation}
	v \Perp nd(v)\bs cl(v) \ver bd(v).
	\end{equation}
\item[(CG)] \textit{global $\cH$-Markovian}, relative to $G$, if for all $A,B,C\subseteq V$ such that $C$ separates $A$ and $B$ in $(G_{ant(A\cup B\cup C)})^m$, the moral graph of the smallest ancestral set containing $A\cup B \cup C$, we have 
$$ A \Perp B \ver C.$$
\end{description}

The factorization in the case of a chain graph involves two parts. Suppose $\{\tau: \tau \in \cD\}$ is the set of chain components of $G$. Then $P$ is said to \textit{factorize} according to $G$ if it has density $f$ that satisfies:

\begin{enumerate}[(i)]
\item $f$ factorizes as in the directed acyclic case:
 $$f(x) = \ds\prod_{\tau\in \cD} f(x_{\tau} \ver x_{pa(\tau)}).$$
\item For each $\tau \in \cD$, $f$ factorizes in the moral graph of $G_{\tau \cup pa(\tau)}$:
$$f(x_{\tau} \ver x_{pa(\tau)}) = Z^{-1}(x_{pa(\tau)})\ds\prod_{h \in \cC} \psi_h(x)$$
where $\cC$ is the set of maximal cliques in $G_{\tau \cup pa(\tau)}^m$, $\psi_h(x)$ depends only on $x_h$ and $$Z^{-1}(x_{pa(\tau)}) =  \ds\int_{\cX_{\tau}} \prod_{h \in \cC} \psi_h(x) \;\mu_{\tau} (dx_{\tau}).$$
\end{enumerate}
If a probability measure $P$ factorizes according to $G$, then we say $P$ satisfies $\cfc$. From arguments analogous to the directed and undirected cases, we have that in general $$\cfc \implies \cgc \implies \clc \implies \cpc.$$
If we assume $\Sint$, then all Markov properties are equivalent.
\begin{theorem}{(\cite{Fry})}\label{thm:chain-Markov-equiv}
Assume that a probability measure $P$ defined on a chain graph $G$ is such that $\Sint$ holds for disjoint subsets of $V(G)$, then 
$$\cfc \iff \cgc \iff \clc \iff \cpc.$$
\end{theorem}

\section{Bayesian Hypergraphs}\label{sec:HBN}

A \textit{Bayesian hypergraph} (BH) is a probabilistic graphical model that represents a set of variables and their conditional dependencies through an acyclic directed hypegraph $\cH$. Hypergraphs contain many more edges than chain graphs. Thus a Bayesian hypergraph is a more general and powerful framework for studying conditional independence relations that arise in various statistical contexts.

\subsection{Markov Properties of Bayesian hypergraphs}

Analogous to chain graph's case, we can define the Markov properties of a Bayesian hypergraph in a variety of ways. Let $\cH$ be a directed acyclic hypergraph with chain components $\{\tau:\tau \in \cD\}$. We say that a probability measure $P$ defined on $\cX = \times_{\alpha \in V(\cH)} \cX_{\alpha}$ is:

\begin{description}

\item[(HP)] \textit{pairwise $\cH$-Markovian}, relative to $\cH$, if for every pair $(v,u)$ of non-adjacent vertices in $\cH$ with $u \in nd(v)$, 
	\begin{equation}
	v \Perp u \ver nd(v) \backslash \{v,u\}.
	\end{equation}

\item[(HL)] \textit{local $\cH$-Markovian}, relative to $\cH$, if for any vertex $v \in V(\cH)$, 
	\begin{equation}
	v \Perp nd(v)\bs cl(v) \ver bd(v).
	\end{equation}
\item[(HG)] \textit{global $\cH$-Markovian}, relative to $\cH$, if for all $A,B,C\subseteq V$ such that $C$ separates $A$ and $B$ in $\lp\partial(\cH_{ant(A\cup B\cup C)})\rp^m$, the moral graph of the (directed) shadow of the smallest ancestral set containing $A\cup B\cup C$, we have 
$$ A \Perp B \ver C.$$


\end{description}

\begin{definition}
A \textit{Bayesian hypergraph} is a triple $(V,\cH, P)$ such that $V$ is a set of random variables, $\cH$ is a DAH on the vertex set $V$ and $P$ is a multivariate probability distribution on $V$ such that the local Markov property, i.e., $\hlc$, holds with respect to the DAH $\cH$.
\end{definition}

Given a Bayesian hypergraph $(V,\cH,P)$, we call $\cH$ the \textit{Bayesian hypergraph structure} or the \textit{underlying DAH} of the Bayesian hypergraph. For ease of reference, we simply use $\cH$ to denote the Bayesian hypergraph. Moreover, for a Bayesian hypergraph $\cH$ whose underlying DAH is a LWF DAH, we call $\cH$ a \textit{LWF Bayesian hypergraph}.

\begin{remark}\label{rk:Markov-equiv}
Observe that by Proposition \ref{prop:shadow-lem} and the definitions of the hypergraph Markov properties, a Bayesian hypergraph has the same pairwise, local and global Markov properties as its shadow, which is a chain graph.
\end{remark}

By Remark \ref{rk:Markov-equiv}, we can derive the following corollaries from results on the Markov properties of chain graphs:

\begin{corollary}\label{cor:Markov_easy}
$$\hgc \implies \hlc \implies \hpc.$$
\end{corollary}

Furthermore, if we assume $\Sint$, then the global, local and pairwise Markov properties are equivalent.
\begin{corollary}\label{cor:BH-Markov-equiv}
Assume that $P$ is such that $\Sint$ holds for disjoint subsets of V. Then $$\hgc \iff \hlc \iff \hpc.$$
\end{corollary}
\begin{proof}
This follows from Remark \ref{rk:Markov-equiv} and Theorem \ref{thm:chain-Markov-equiv}.
\end{proof}

Given a chain graph $G$, a triple $(\alpha, B, \beta)$ is a \textit{complex}\footnote{or \textit{U-structure} \cite{CWm}} in $G$ if $B$ is a connected subset of a chain component $\tau$, and $\alpha, \beta$ are two non-adjacent vertices in $bd(\tau) \cap bd(B)$. Moreover, $(\alpha, B, \beta)$ is a \textit{minimal} complex if $B= B'$ whenever $B'$ is a subset of $B$ and $(\alpha, B',\beta)$ is a complex. Frydenberg \cite{Fry} showed that two chain graphs have the same Markov properties if they have the same underlying undirected graph and the same minimal complexes. In the case of a Bayesian hypergraph, by Remark \ref{rk:Markov-equiv} and the result on the Markov equivalence of chain graphs, we obtain the following conclusion on the Markov equivalence of Bayesian hypergraphs.

\begin{corollary}
Two Bayesian hypergraphs have the same Markov properties if their shadows are Markov equivalent, i.e., their shadows have the same underlying undirected graph and the same minimal complexes.
\end{corollary}

\subsection{Factorization according to Bayesian hypergraphs}\label{sec:BH-fact}

The factorization of a probability measure $P$ according to a Bayesian hypergraph is similar to that of a chain graph. Before we present the factorization property, let us introduce some additional terminology. 

Given a DAH $\cH$, we use $\cH^u$ to denote the undirected hypergraph obtained from $\cH$ by replacing each directed hyperedge $h = (A,B)$ of $\cH$ into an undirected hyperedge $A\cup B$. Given a family of sets $\cF$, define a partial order $(\cF, \leq)$ on $\cF$ such that for two sets $A,B \in \cF$, $A \leq B$ if and only if $A\subseteq B$. Let $\cM(\cF)$ denote the set of maximal elements in $\cF$, i.e., no element in $\cM(\cF)$ contains another element as subset. When $\cF$ is a set of directed hyperedges, we abuse the notation to denote $\cM(\cF) = \cM(\cF^u)$. 

Let $\cH$ be a directed acyclic hypergraph and $\{\tau:\tau\in \cD\}$ be its chain components.
Assume that a probability distribution $P$ has a density $f$, with respect to some product measure $\mu = \times_{\alpha \in V}  \;\mu_{\alpha}$ on $\cX = \times_{\alpha \in V} \cX_{\alpha}$. 
Now we say a probability measure $P$ \textit{factorizes} according to $\cH$ if it has density $f$ such that 
\begin{enumerate}[(i)]
\item $f$ factorizes as in the directed acyclic case:
 \begin{equation}\label{eq:fact1}
     f(x) = \ds\prod_{\tau\in \cD} f(x_{\tau} \ver x_{pa(\tau)}).
 \end{equation}

\item For each $\tau \in \cD$, define $\cH_{\tau}^*$ to be the subhypergraph of $\cH_{\tau \cup pa(\tau)}$ containing all edges $h$ in $\cH_{\tau \cup pa(\tau)}$ such that $H(h) \subseteq \tau$. 

\begin{equation}\label{eq:fact}
   f(x_{\tau} \ver x_{pa(\tau)}) = \ds\prod_{h \in \cM(\cH_{\tau}^*)} \psi_h(x).
\end{equation}
where $\psi_h$ are non-negative functions depending only on $x_h$ and $\ds\int_{\cX_{\tau}} \prod_{h \in \cM(\cH_{\tau}^*)} \psi_h(x) \mu_{\tau} (dx_{\tau}) = 1$.
Equivalently, we can also write $f(x_{\tau} \ver x_{pa(\tau)})$ as
\begin{equation}\label{eq:fact-alt}
f(x_{\tau} \ver x_{pa(\tau)}) = Z^{-1}(x_{pa(\tau)})\ds\prod_{h \in \cM(\cH_{\tau}^*)} \psi_h(x),
\end{equation}
where $Z^{-1}(x_{pa(\tau)}) =\ds\int_{\cX_{\tau}} \prod_{h \in \cM(\cH_{\tau}^*)} \psi_h(x) \mu_{\tau} (dx_{\tau}) $.

\end{enumerate}

\begin{remark}
Note that although (LWF) Bayesian hypergraphs are generalizations of of Bayesian networks and LWF chain graph models, the underlying graph structures that represent the same factorizations may differ. Hence the underlying graph structures of Bayesian networks and chain graph do not directly migrate to Bayesian hypergraphs.

\begin{figure}[htb]
	\begin{center}
	\resizebox{2.5cm}{!}{\input{images/BH_altfact/BH_altfact_ex2.tikz}}
	\end{center}
	\caption{A simple Bayesian hypergraph $\cH$.}
	\label{f:PC-GC}
\end{figure}
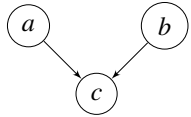

We will illustrate with an example. Consider the graph in Figure \ref{f:PC-GC}, which can be interpreted as a chain graph structure $G$ or a Bayesian hypergraph structure $\cH$. Note that the factorizations, under the two interpretations, are different. In particular, the factorization, according to $G$, is  
$$f_G(x) = f(x_a) f(x_b) \psi_{abc}(x)$$
for some non-negative functions $\psi_{abc}$. On the other hand, the factorization, according to $\cH$, is
$$f_{\cH}(x) = f(x_a) f(x_b) \psi_{ac}(x) \psi_{bc}(x)$$ for some non-negative functions $\psi_{ac}, \psi_{bc}.$

\end{remark}


\begin{figure}[htb]
	\begin{center}
		\begin{minipage}{.2\textwidth}
			\resizebox{2.5cm}{!}{\input{images/CG_fact_ex2.tikz}}
		\end{minipage}
		\hspace{1cm}
		\begin{minipage}{.2\textwidth}
			\resizebox{2.5cm}{!}{\input{images/BH_fact_ex2.tikz}}
		\end{minipage}
		\caption{(1) a chain graph $G$; (2) a Bayesian hypergraph $\cH$.}
		\label{f:further-fact}
	\end{center} 
\end{figure}
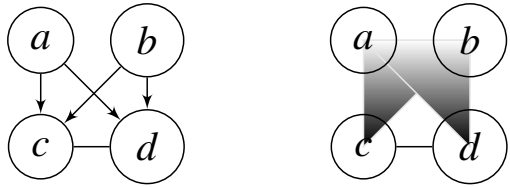

\begin{remark}\label{re:mem_save}
One of the key advantages of Bayesian hypergraphs is that they allow much finer factorizations of probability distributions compared to chain graph models. We will illustrate with a simple example in Figure \ref{f:further-fact}. Note that in Figure \ref{f:further-fact} (1), the factorization according to $G$ is 
\begin{align*} 
f(x) &= f(x_a) f(x_b) f(x_{cd} \ver x_{ab})\\
     &= f(x_a) f(x_b) \psi_{abcd}(x)
\end{align*}
In Figure \ref{f:further-fact} (2), the factorization according to $\cH$ is 
\begin{align*}
f(x) &= f(x_a) f(x_b) f(x_{cd} \ver x_{ab})\\
     &= f(x_a) f(x_b) \psi_{abc}(x) \psi_{abd}(x) \psi_{cd}(x)
\end{align*}

Note that although $G$ and $\cH$ have the same global Markov properties, the factorization according to $\cH$ is one step further compared to the factorization according to $G$. Suppose each of the variables of $\{a,b,c,d\}$ can take $k$ values. Then the factorization according to $G$ will require a conditional probability table of size $k^4$ while the factorization according to $\cH$ only needs a table of size $\Theta(k^3)$ asymptotically. Hence, a Bayesian hypergraph model allows much finer factorizations and thus achieves higher memory efficiency.
\end{remark}

\begin{remark}\label{rm:fact-all}
We remark that the factorization formula defined in \eqref{eq:fact} is in fact the most general possible in the sense that it allows all possible factorizations of a probability distribution admitted by a DAH. In particular, given a Bayesian hypergraph $\cH$ and one of its chain components $\tau$, the factorization scheme in \eqref{eq:fact} allows a distinct function for each maximal subset of $\tau \cup pa_{\cD}(\tau)$ that intersects $\tau$ ($pa_{\cD}$ is the parent of $\tau$ in the canonical DAG of $\cH$). For each subset $S$ of $\tau\cup pa_{\cD}(\tau)$ that does not intersect $\tau$, 
recall that the factorization in \eqref{eq:fact} can be rewritten as follows:
$$f(x_{\tau} \ver x_{pa(\tau)}) = \lp \ds\prod_{h \in \cM(\cH_{\tau}^*)} \psi_h(x) \rp \bigg/ \lp \ds\int_{\cX_{\tau}} \prod_{h \in \cM(\cH_{\tau}^*)} \psi_h(x) \mu_{\tau} (dx_{\tau})\rp.$$
Observe that $\psi_S(x)$ is a function that does not depend on values of variables in $\tau$. Hence $\psi_S(x)$ can be factored out from the integral above and cancels out with itself in $f(x_{\tau} \ver x_{pa(\tau)})$. Thus, the factorization formula in \eqref{eq:fact} or \eqref{eq:fact-alt} in fact allows distinct functions for all possible maximal subsets of $\tau \cup pa_{\cD}(\tau)$.

\end{remark}

\begin{table}[ht]
	\centering
	\begin{tabular}{| M{40mm} | M{30mm} | M{40mm} | M{30mm} | } 
		\hline
		
		\centering{Factorization} & \centering{BH representation} & \centering{Factorization} & \centering{BH representation} \tabularnewline
		\hline  
		
		\centering{$f(x) = f(x_a)f(x_b) \psi_{abc}(x)$} & \centering{\resizebox{2cm}{!}{\input{images/BH_altfact/BH_altfact_ex1.tikz}}} & 
		\centering{$f(x) = f(x_{ab}) \psi_{abc}(x)$} & 
		\centering{\resizebox{2cm}{!}{\input{images/BH_altfact/BH_altfact_ex1B.tikz}}} \tabularnewline
		\hline
		
		\centering{$f(x) = f(x_a)f(x_b) \psi_{ac}(x) \psi_{bc}(x)$} & \centering{\resizebox{2cm}{!}{\input{images/BH_altfact/BH_altfact_ex2.tikz}}} & 
		\centering{$f(x) = f(x_{ab}) \psi_{ac}(x) \psi_{bc}(x)$} & 
		\centering{\resizebox{2cm}{!}{\input{images/BH_altfact/BH_altfact_ex2B.tikz}}} \tabularnewline
		\hline
		
		
		\centering{$f(x) = f(x_a)f(x_b) \psi_{ac}(x)$} &
		\centering{\resizebox{2cm}{!}{\input{images/BH_altfact/BH_altfact_ex4.tikz}}} & 
		\centering{$f(x) = f(x_{ab}) \psi_{ac}(x)$} & 
		\centering{\resizebox{2cm}{!}{\input{images/BH_altfact/BH_altfact_ex4B.tikz}}} \tabularnewline
		\hline
		
		\centering{$f(x) = f(x_a)f(x_b) f(x_c)$} &
		\centering{\resizebox{2cm}{!}{\input{images/BH_altfact/BH_altfact_ex5.tikz}}} & 
		\centering{$f(x) = f(x_{ab})f(x_c)$} & 
		\centering{\resizebox{2cm}{!}{\input{images/BH_altfact/BH_altfact_ex5B.tikz}}} \tabularnewline
		\hline
		
		\centering{$f(x) = \psi_{ac}(x)\psi_{bc}(x)$} &
		\centering{\resizebox{2cm}{!}{\input{images/BH_altfact/BH_altfact_ex6.tikz}}} & 
		\centering{$f(x) = f(x_c)\psi_{ac}(x)\psi_{ab}(x)$} & 
		\centering{\resizebox{2cm}{!}{\input{images/BH_altfact/BH_altfact_ex6B.tikz}}} \tabularnewline
		\hline
		
		\centering{$f(x) =f(x_c) \psi_{ac}(x) \psi_{bc}(x)$} &
		\centering{\resizebox{2cm}{!}{\input{images/BH_altfact/BH_altfact_ex7.tikz}}} & 
		\centering{$f(x) = f(x_c)\psi_{ab}(x)\psi_{ac}(x)\psi_{bc}(x)$} & 
		\centering{\resizebox{2cm}{!}{\input{images/BH_altfact/BH_altfact_ex7B.tikz}}} \tabularnewline
		\hline

	\end{tabular}
	\caption{Factorizations and corresponding BH representations}\label{t3}
\end{table}

Table \ref{t3} lists some factorizations of three random variables and the corresponding BH representation. Entry 1 (top left) corresponds to a three-node Bayesian network: an uncoupled converging connection (unshielded collider) at $c$. Entry 3 (below entry 1) corresponds to a three-node Bayesian network like the one in entry 1, with the constraint that the conditional probability table factorizes as, for example, in a Noisy-OR functional dependence and, more generally, in a situation for which compositionality holds, such as MIN, MAX, or probabilistic sum \cite{Pearl88, Hajek92, JN07}.  Graphical modeling languages should capture assumptions graphically in a transparent and explicit way, as opposed to hiding them in tables or functions.  By this criterion, the Bayesian hypergraph of entry 3 shows the increased power of our new PGM with respect to Bayesian networks and chain graphs.

	\begin{figure}[ht]
		\centering
		\includegraphics[width=11cm]{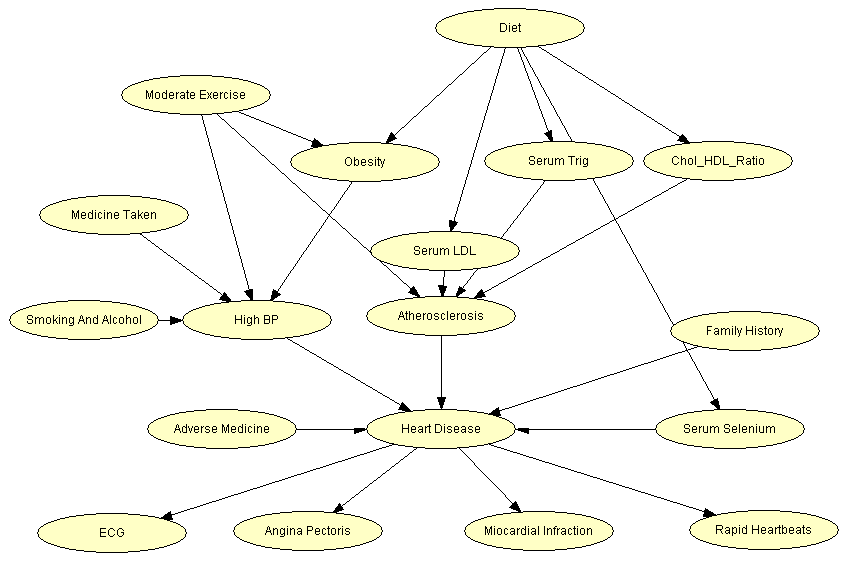}
		\caption{A model of heart disease} \label{Fig:heart}
	\end{figure}

For a detailed example of Noisy-OR functional dependence, consider the (much simplified) heart disease model of \cite{GV00}, shown in Figure  \ref{Fig:heart}, and the family of nodes Obesity (O, with values Yes, No), Diet (D, with values Bad, Good), and Moderate Exercise (M, with values Yes, No).  The Noisy-OR model is used to compute the conditional probability of O given M and D. Good diet prevents obesity, except when an inhibiting mechanism prevents that with probability $q_{D \to O}$; moderate exercise prevents obesity except when an inhibiting mechanism prevents that with probability $q_{M \to O}$. The inhibiting mechanisms are independent, and therefore the probability of being obese given both a good diet and moderate exercise is $1 - q_{D \to O} q_{M \to O}$.  Equivalently, the probability of not being obese given both a good diet and moderate exercise is $q_{D \to O} q_{M \to O}$. If we consider a situation with only the variables just described, the joint probability of Diet, Moderate Exercise, and Obesity factorizes exactly as in the Bayesian hypergraph of entry 3, with the caution that only half of the entries in the joint probability table are computed directly; the others are computed by the complement to one. 

	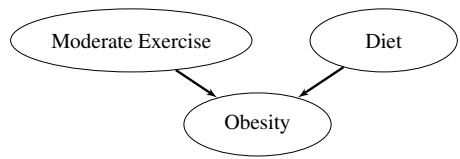
\begin{figure}[htb]
	\begin{center}
			\resizebox{6cm}{!}{\input{images/obesity.tikz}}
		\caption{An example of Noisy-OR: obesity.}
		\label{f:obesity}
	\end{center} 
\end{figure}

Similarly, entry 2 corresponds to a three node chain graph, while entry 4 may be used to model a situation in which variables $a$ and $b$ are related by being effects of a common latent cause, while the mechanisms by which they, in turn, affect variable $c$ are causally independent.  While such a situation may be unusual, it is notable that it can be represented graphically in Bayesian hypergraphs. Therefore, the Bayesian hypergraph of entry 4 shows the increased power of our new PGM with respect to Bayesian networks and chain graphs.

	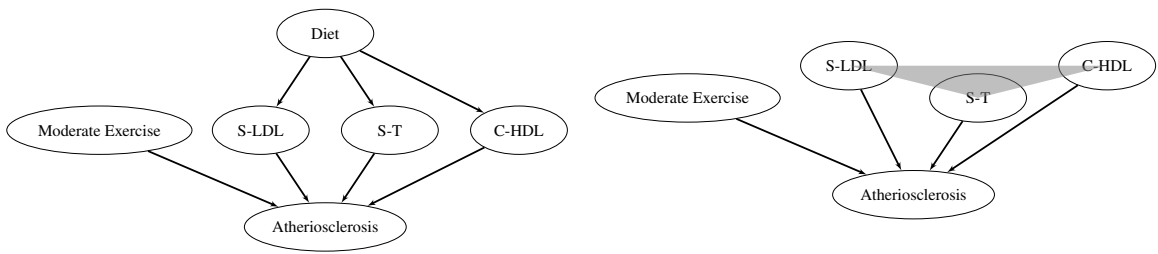
\begin{figure}[htb]
	
		\begin{minipage}{.5\textwidth}
		\begin{center}
			\resizebox{7.5cm}{!}{\input{images/AtherA.tikz}}
		\end{center}
		\end{minipage}
		\begin{minipage}{.5\textwidth}
		\begin{center}
		   \resizebox{7.5cm}{!}{\input{images/AtherB.tikz}}
		\end{center}
		\end{minipage}
		\caption{An example of Noisy-OR: Atheriosclerosis.}
		\label{f:Ather}
	 
\end{figure}

For a detailed example, consider again the model shown in Figure  \ref{Fig:heart} and, this time, the structure in which Moderate Exercise, Serum LDL (S-LDL), Serum Triglicerides (S-T), and Cholesterol HDL (C-HDL) Ratio are parents (possible causes) of Atheriosclerosis, and Diet is a parent of S-LDL, S-T, and C-HDL.  As in the previous example, the Noisy-OR assumption is made, and therefore, after marginalization of Diet, the computation of the joint probability of Moderate Exercise, S-LDL, S-T, C-HDL, and Atheriosclerosis factorizes as in an entry 4, with a slight generalization due to the presence of four parents instead of two.  As in entry 4, the parents (causes) are not marginally independent, due to their common dependence on Diet, but the conditional probability of the effect decomposes multiplicatively.

Moreover, as illustrated in Remark \ref{re:mem_save} and Table \ref{t3}, a Bayesian hypergraph enables much finer factorization than a chain graph. In the factorization w.r.t. a chain graph $G$ with chain components $\{\tau: \tau\in \cD\}$, $f(x_{\tau} \ver x_{pa(\tau)})$ is only allowed to be further factorized based on the maximal cliques in the moral graph of $G_{\tau\cup pa(\tau)}$, which is rather restrictive. In comparison, a Bayesian hypergraph $\cH$ allows factorization based on the maximal elements in all subsets of the power set of $\tau \cup pa_{\cD}(x)$. Finer factorizations have the advantage of memory saving in terms of the size of the probability table required. Moreover, factorizations according to Bayesian hypergraphs can be obtained directly from reading off the hyperedges instead of having to search for all maximal cliques in the moral graph (in the chain graph's case). Hence, Bayesian hypergraphs enjoy an advantage in heuristic adequacy as well as representational adequacy.

Next, we investigate the relationship between the factorization property and the Markov properties of Bayesian hypergraphs.
\begin{proposition}\label{prop:fact-Markov}
Let $P$ be a probability measure with density $f$ that factorizes according to a DAH $\cH$. Then
$$\hfc \implies \hgc \implies \hlc \implies \hpc.$$
\end{proposition}
\begin{proof}
It suffices to show $\hfc \implies \hgc$ since the other implications are proven in Corollary \ref{cor:Markov_easy}. Let $A,B,C \subseteq V(\cH)$ such that $C$ separates $A$ and $B$ in $G = \lp\partial(\cH_{ant(A\cup B\cup C)})\rp^m$. Let $\tilde{A}$ be the connectivity components in $G\bs C$ containing $A$ and let $\tilde{B} = V \bs (\tilde{A} \cup C)$. Note that in $\lp\partial(\cH_{ant(A\cup B\cup C)})\rp^m$, every hyperedge $h = (T,H)$ becomes a complete graph on the vertex set $T\cup H$ because of moralization.
Observe that since $C$ separates $A$ and $B$ in $G$, for every hyperedge $h = (T,H)$, $T\cup H$ is either a subset of $\tilde{A}\cup C$ or $\tilde{B} \cup C$. Let $\cH' = \cH_{ant(A\cup B\cup C)}$ and $\{\tau:\tau \in \cD'\}$ be the chain components of $\cH'$. For each $\tau \in \cD'$, define $\cH_{\tau}^*$ to be the subhypergraph of $\cH'_{\tau \cup pa(\tau)}$ containing all edges $h$ in $\cH'_{\tau \cup pa(\tau)}$ such that $H(h) \subseteq \tau$. 
We then obtain from the $\hfc$ property that 
\begin{align*}
    f_{\cH'}(x) &= \ds\prod_{\tau \in \cD'} \ds\prod_{h \in \cM(\cH_{\tau}^*)} \psi_h(x).  \\
               &=  \phi_1 (x_{\tilde{A} \cup C}) \phi_2(x_{\tilde{B}\cup C}).
\end{align*} 
for some non-negative functions $\phi_1, \phi_2$. By integrating over the chain components not in $ant(A\cup B\cup C)$, it follows that 
$$f(x) = \psi_1 (x_{\tilde{A} \cup C}) \psi_2(x_{\tilde{B}\cup C}).$$for some non-negative functions $\psi_1, \psi_2$. Hence, we have that
$$\tilde{A}\Perp \tilde{B} \ver C.$$
By (S2: Decomposition) property of conditional independences, it follows that $A\Perp B \ver C$.
\end{proof}

\begin{remark}
Due to the generality of factorizations according to Bayesian hypergraphs, the reverse direction of the implication $\hfc \implies \hgc$ in Proposition \ref{prop:fact-Markov} is generally not true. We will illustrate with the following example.

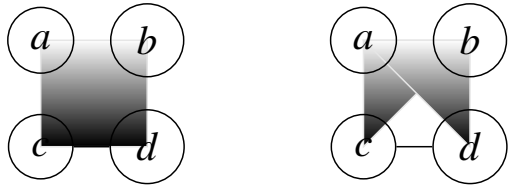
\begin{figure}[htb]
	\begin{center}
		\begin{minipage}{.2\textwidth}
			\resizebox{2.5cm}{!}{\input{images/BH_fact_ex2A.tikz}}
		\end{minipage}
		\hspace{1cm}
		\begin{minipage}{.2\textwidth}
			\resizebox{2.5cm}{!}{\input{images/BH_fact_ex2.tikz}}
		\end{minipage}
		\caption{Two Bayesian hypergraphs $\cH_1$ (left), $\cH_2$ (right) with the same global Markov properties but different forms of factorizations.}
		\label{f:FC-PC}
	\end{center} 
\end{figure}
Consider the two Bayesian hypergraphs $\cH_1$ and $\cH_2$ in Figure \ref{f:FC-PC}. Note that they have the same global Markov properties since the shadows of $\cH_1$ and $\cH_2$ are the same. However the factorizations according to $\cH_1$ and $\cH_2$ are different. If we let $f_1, f_2$ denote the factorizations represented by $\cH_1$ and $\cH_2$, then 
\begin{align*} 
f_1(x) &= f_1(x_a) f_1(x_b) f_1(x_{cd} \ver x_{ab})\\
     &= f_1(x_a) f_1(x_b) \psi_{abcd}(x)
\end{align*}
while
\begin{align*}
f_2(x) &= f_2(x_a) f_2(x_b) f_2(x_{cd} \ver x_{ab})\\
     &= f_2(x_a) f_2(x_b) \psi_{abc}(x) \psi_{abd}(x) \psi_{cd}(x)
\end{align*}
This shows that $\hgc$ does not generally imply $\hfc$. \end{remark}

\begin{remark}\label{counting}
We give another combinatorial argument on why in general $\hfc$ does not imply $\hgc$. We claim that the number of possible forms of factorizations admitted by Bayesian hypergraphs is much more than the number of conditional independence statements over the same set of variables. First, observe that the number of conditional independence statements on $n$ variables is upper bounded by the number of ways to partition $n$ elements into four disjoint sets $A,B,C,D$. Each such partition induces a conditional statement $A \Perp B \ver C$ and $D$ is the set of unused variables. There are $4^n$ ways to partition $n$ distinct elements into four ordered pairwise disjoint sets. Hence there are at most $4^n$ conditional independence statements on $n$ variables.

On the other hand, we give a simple lower bound on the number of directed acyclic hypergraphs by simply counting the number of directed acyclic hypergraphs $\cH$ whose vertex sets can be partitioned into two sets $A,B$ such that $|A|=|B|= n/2$ and every fully directed edge has its tail only from $A$ and its head only from $B$. Observe that there are $2^{n/2}$ subsets of $A$ and $B$ respectively. By Sperner's theorem \cite{Sperner}, the largest number of subsets of $A$ none of which contain any other is upper bounded by $\binom{n/2}{n/4}$. The same holds for $B$. Hence there are at least $\binom{n/2}{n/4}^2$ possible directed hyperedges such that when viewed as undirected hyperedge, no edge contains any other as subset. Therefore, there are at least $$2^{\binom{n/2}{n/4}^2} = \Theta\lp 2^{\frac{2^{n+2}}{\pi n}}\rp$$
distinct factorizations admitted by DAHs whose directed edges have their tails only from $A$ and their heads only from $B$. Note that this number is much less than the total number of distinct factorizations admitted by DAHs, but is already much bigger than $4^n$, which is the upper bound on the number of conditional independence statements on $n$ variables. Hence, there are many more factorizations allowed by Bayesian hypergraphs than the number of conditional independence statements on $n$ variables, which suggest that $\hgc$ does not imply $\hfc$ in general.
\end{remark}

\subsection{Comparison between LWF chain graph and LWF Bayesian hypergraph}

\begin{theorem}\label{thm:HG-equiv}
Let $G$ be a chain graph and $\cH$ be its canonical (LWF) DAH. We show that a probability measure $P$ satisfies the following:
\begin{enumerate}[(i)]

\item\label{p-equiv} $P$ is pairwise $G$-Markovian if and only if $P$ is pairwise $\cH$-Markovian.
\item\label{l-equiv} $P$ is local $G$-Markovian if and only if $P$ is local $\cH$-Markovian.
\item\label{g-equiv} $P$ is global $G$-Markovian if and only if $P$ is global $\cH$-Markovian.

\end{enumerate}
\end{theorem}
\begin{proof}
By Lemma \ref{lem:comp}, $nb_G(v) = nb_{\cH}(v)$, $pa_G(v) = pa_{\cH}(v)$. Hence the same equality holds for $nd_G(v),bd_G(v),cl_G(v)$, which gives us \ref{p-equiv} and \ref{l-equiv} by definition of the Markov properties.
\ref{g-equiv} results from the fact that for all $A,B,C\subseteq V(G)= V(\cH)$, $G_{ant(A\cup B\cup C)} = \partial(\cH_{ant(A\cup B\cup C)})$.
\end{proof}

\begin{theorem}\label{thm:factor-equiv}
Let $G$ be a chain graph and $\cH$ be its canonical LWF DAH. Then a probability measure $P$ with density $f$ factorizes according to $G$ if and only if $f$ factorizes according to $\cH$.
\end{theorem}

\begin{proof}
Note that by Lemma \ref{lem:comp}, $G$ and $\cH$ have the same set of chain components $\{\tau: \tau\in \cD\}$. It suffices to show for every $\tau \in \cD$, there exists a bijective map $\phi$ from the set of maximal edges in $(\cH^*_{\tau})^u$ to the set of maximal cliques in $(G_{\tau \cup pa(\tau)})^m$ such that for each maximal edge $h$ in $(\cH^*_{\tau})^u$, $\phi(h) = h$.
For ease of reference, let $\cH' = (\cH^*_{\tau})^u$ and let $G' = (G_{\tau \cup pa(\tau)})^m$. Define $\phi(h) = h$. We need to show two things: (1) for every maximal edge $h$ in $\cH'$, $h$ induces a maximal clique in $G'$; (2) for every maximal clique $h$ in $G'$, $h$ is a maximal edge in $\cH'$. 

We first show (1). Suppose that $h$ is a maximal edge in $\cH'$. Clearly, $h$ induces a clique in $G'$ because of the moralization. Suppose for the sake of contradiction that $h$ is not maximal in $G'$, i.e. there is a maximal clique $h'$ in $G'$ such that $h\subsetneq h'$. Let $h' = A\cup B$ where $A \subseteq pa(\tau)$ and $B \subseteq \tau$. There are two cases:
\begin{description}
\item Case 1: $A = \emptyset$ or $B = \emptyset$.

Note that $B$ cannot be an empty set since $h$ is an edge in $\cH'$  and every edge in $\cH' = (\cH^*_{\tau})^u$ intersects $\tau$ by definition. 
If $A = \emptyset$, then $h'$ is a maximal clique in $\tau$. By Phase I of the construction, $h'$ either is a hyperedge in $\cH'$ or is contained in the head of a hyperedge. In either case, since $h\subsetneq h'$, it contradicts that $h$ is a maximal edge in $\cH'$. 

\item Case 2: $A \neq \emptyset$ and $B\neq \emptyset$. 

Since $A\cup B$ induces a maximal clique in $G'$, it follows that for every $a \in A, b\in B$, $a \in pa(b)$. Hence $B$ the common children of some elements in $A$. Recall that in Phase I of our construction, for every $v$, $(\{v\},K_v)$ is an hyperedge in $\cH$ where $K_v$ is a maximal clique in the children of $v$ in $G$. Hence there exists $\cF \subseteq E(\cH^*_{\tau})$ such that $B \subseteq \cap_{h\in \cF} H(h)$. By maximality of $h'$, $B = \cap_{h\in \cF} H(h)$. Now by our construction in Phase II, there exists a hyperedge
$$h'' = \Bigg( \ds\bigcup_{\substack{h \in E(\cH^*_{\tau}))\\B\subseteq H(h)}} T(h), B \Bigg)  \in E(\cH^*_{\tau}).$$ 
Since every element in $A$ is a parent of every element in $B$, it follows that 
$$A\subseteq \ds\bigcup_{\substack{h \in E(\cH^*_{\tau}))\\B\subseteq H(h)}} T(h).$$
By maximality of $A$, it follows that 
$$h\subsetneq h' = h'' \in E(\cH^*_{\tau}).$$
which contradicts the maximality of $h$ again.
\end{description}

Hence in both cases, we obtain by contradiction that $h$ induces a maximal clique in $G'$.

It remains to show (2). Suppose $h$ induces a maximal clique in $G'$. Observe that every hyperedge in $\cH'$ induces a clique in $G'$. Similar logic and case analysis above apply and it is not hard to see that $h$ is a maximal edge in $\cH'$. We will leave the details to the reader.
\end{proof}

\noindent\textit{Example.}
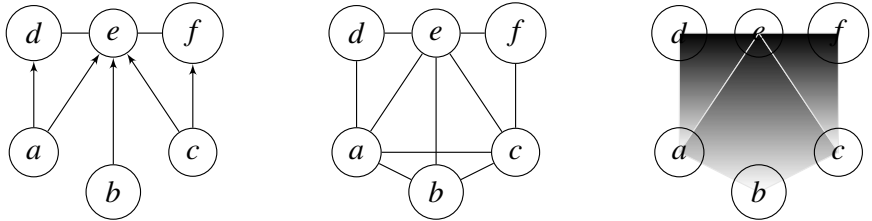
\begin{figure}[htb]
	\begin{center}
		\begin{minipage}{.2\textwidth}
		\resizebox{3cm}{!}{\input{images/CG_ex3.tikz}}
		\end{minipage}
		\hspace{1cm}
		\begin{minipage}{.2\textwidth}
		\resizebox{3cm}{!}{\input{images/CG_ex3_moral.tikz}}
		\end{minipage}
		\hspace{1cm}
		\begin{minipage}{.2\textwidth}
		\resizebox{3cm}{!}{\input{images/HBN_ex3_fact.tikz}}
		\end{minipage}
		\caption{(1) a simple chain graph $G$; (2) The moral graph $G^m$ of $G$; (3) $\cM(\cH)$ where $\cH$ is the cononical LWF DAH of $G$.}
		\label{f:factor}
	\end{center} 
\end{figure}

In Figure \ref{f:factor}, both $G$ and its canonical LWF DAH $\cH$ have chain components $\{\{a\},\{b\},\{c\},\{d,e,f\}\}$.
Figure \ref{f:factor} (2) shows the moral graph $G^m$ of $G$. The maximal cliques in $G^m$ are $\{ade, abce, cef\}$. Thus, by the factorization property of LWF chain graphs, we have that a probability measure $P$ with density $f$ that factorizes according to $G$ satisfies 
\begin{align*}
f(x) &= f(x_a) f(x_b) f(x_c) f(x_{d,e,f} \ver x_{a,b,c})\\
     &= f(x_a) f(x_b) f(x_c) \psi_{ade}(x) \psi_{abce}(x)  \psi_{cef}(x).
\end{align*}
Figure \ref{f:factor} (3) gives the undirected hypergraph with edge set $\cM(\cH)$. Observe that $\cM(\cH)$ has the same members as the set of maximal cliques in $G^m$. Hence by the factorization property of Bayesian hypergraphs, they admit the same factorization.

\section{Intervention in Bayesian hypergraphs}\label{sec:intervention}

Formally, intervention in Bayesian hypergraphs can be defined analogously to intervention in LWF chain graphs \cite{LR}. In this section, we give graphical procedures that are consistent with the intervention formulas for chain graphs (Equation \eqref{intervention-CG}, \eqref{intervention-LWF}) and for Bayesian hypergraphs (Equation \eqref{bh:intervention1}, \eqref{bh:intervention2}). Before we present the details, we need some additional definitions and tools to determine when factorizations according to two chain graphs or DAHs are equivalent in the sense that they could be written as products of the same type of functions (functions that depend on same set of variables). We say two chain graphs $G_1, G_2$ admit the same \textit{factorization decomposition} if for every probability density $f$ that factorizes according to $G_1$, $f$ also factorizes according to $G_2$, and vice versa. Similarly, two  DAHs $\cH_1, \cH_2$ admit the same \textit{factorization decomposition} if for every probability density $f$ that factorizes according to $\cH_1$, $f$ also factorizes according to $\cH_2$, and vice versa.

\subsection{Factorization equivalence and intervention in  chain graphs}

In this subsection, we will give graphical procedures to model intervention based on the formula introduced by Lauritzen and Richardson in \cite{LR}. Let us first give some background.
In many statistical context, we would like to modify the distribution of a variable $Y$ by intervening externally and forcing the value of another variable $X$ to be $x$. This is commonly refered as \textit{conditioning by intervention} or \textit{conditioning by action} and denoted by $Pr(y \| x)$ or $Pr(y \ver X \leftarrow x)$.
Other expressions such as $Pr(Y_x = y), P_{man(x)}(y), set(X=x), X = \hat{x}$ or $do(X = x)$ have also been used to denote intervention conditioning (Neyman \cite{Neyman}; Rubin \cite{Rudin}; Spirtes et al. \cite{Spirtes}; Pearl \cite{Pearl93, Pearl95, Pearl20}). 

Let $G$ be a chain graph with chain components $\{\tau: \tau \in \cD\}$. Moreover, assume further that a subset $A$ of variables in $V(G)$ are set such that for every $a\in A$, $x_a = a_0$.
Lauritzen and Richardson, in \cite{LR}, generalized the conditioning by intervention formula for DAGs and gave the following formula for intervention in chain graphs (where it is understood that the probability of any configuration of variables inconsistent with the intervention is zero). A probability density $f$ factorizes according to $G$ (with $A$ intervened) if 

\begin{equation}\label{intervention-CG}
f(x \| x_A) = \ds\prod_{\tau\in \cD} f(x_{\tau \backslash A} \ver x_{pa(\tau)}, x_{\tau \cap A}).
\end{equation}
Moreover, for each $\tau \in \cD$,

\begin{equation}\label{intervention-LWF}
    f(x_{\tau \backslash A} \ver x_{pa(\tau)}, x_{\tau \cap A}) = Z^{-1}(x_{pa(\tau)}, x_{\tau \cap A}) \ds\prod_{h \in \cC} \psi_h(x)
\end{equation}
where $\cC$ is the set of maximal cliques in $(G_{\tau \cup pa(\tau)})^m$ and $Z^{-1}(x_{pa(\tau)}, x_{\tau \cap A}) = \ds\int_{\cX_{\tau \backslash A}} \ds\prod_{h \in \cC} \psi_h(x) \mu_{\tau\backslash A} (dx_{\tau\backslash A})$.

\begin{definition}
$G_1$ and $G_2$ be two chain graphs.
Given a subset $A_1 \subseteq V(G_1)$ and $A_2 \subseteq V(G_2)$,
we say $(G_1, A_1)$ and $(G_2, A_2)$ are \emph{factorization-equivalent}\footnote{This term was defined for a different purpose in \cite{Studeny09}.} if they become the same chain graph after removing from $G_i$  all vertices in $A_i$ together with the edges incident to vertices in $A_i$ for $i\in \{1,2\}$. Typically, $A_i$ is a set of constant variables in $V(G_i)$ created by intervention.
\end{definition}

\begin{theorem}\label{thm: CG_factorization-equivalence}
Let $G_1$ and $G_2$ be two chain graphs defined on the same set of variables $V$. Moreover a common set of variables $A$ in $V$ are set by intervention such that for every $a \in A$, $x_a = a_0$. If $(G_1,A)$ and $(G_2,A)$ are factorization-equivalent, then $G_1$ and $G_2$ admit the same factorization decomposition.
\end{theorem}
\begin{proof}
Let $G_0$ be the chain graph obtained from $G_1$ by removing all vertices in $A$ and the edges incident to $A$. It suffices to show that $G_1$ and $G_2$ both admit the same factorization decomposition as $G_0$. Let $\cD_1$, $\cD_0$ be the set of chain components of $G_1$ and $G_0$ respectively. Let $\tau \in \cD_1$ be an arbitrary chain component of $G_1$. By the factorization formula in \eqref{intervention-LWF}, it follows that 
 $$f(x_{\tau \backslash A} \ver x_{pa(\tau)}, x_{\tau \cap A}) = Z^{-1}(x_{pa(\tau)}, x_{\tau \cap A}) \ds\prod_{h \in \cC} \psi_h(x)$$
 where $\cC$ is the set of maximal cliques in $(G_{\tau \cup pa(\tau)})^m$ and $Z^{-1}(x_{pa(\tau)}, x_{\tau \cap A}) = \ds\int_{\cX_{\tau \backslash A}} \ds\prod_{h \in \cC} \psi_h(x) \mu_{\tau\backslash A} (dx_{\tau\backslash A})$.
 Notice that for any maximal clique $h_1 \in \cC$ such that $h_1\cap A = \emptyset$, $h_1$ is also a clique in $(G_0[\tau \backslash A])^m$. For $h_1 \in \cC$ with $h_1\cap A \neq \emptyset$, there are two cases:
 \begin{description}
     \item Case 1: $(h_1\cap \tau) \backslash A \neq \emptyset$. In this case, observe that $h_1\backslash A$ is also a clique in $(G_0[\tau \backslash A])^m$, thus is contained in some maximal clique $h'$ in $(G_0[\tau \backslash A])^m$. Since all variables in $A$ are pre-set as constants, it follows that $\psi_{h_1}(x)$ also appears in  a factor in the factorization of $f$ according to $G_0$.
     \item Case 2: $h_1\cap \tau \subseteq A$. In this case, note that $h_1\cap \tau$ is disjoint with $\tau \backslash A$. Hence $\psi_{h_1}(x)$ appears as a factor independently of $x_{\tau \backslash A}$ in both $Z^{-1}(x_{pa(\tau)}, x_{\tau \cap A})$ and $\ds\prod_{h \in \cC} \psi_h(x)$, which cancels out with itself. 
 \end{description}
 Thus it follows that  every probability density $f$ that factorizes according to $G_1$ also factorizes according to $G_0$.
 On the other hand, it is easy to see that for every $\tau' \in \cD_0$ and every maximal clique $h'$ in $(G_0[\tau'])^m$, $h'$ is contained in some maximal clique $h$ in $(G_1[\tau])^m$ for some $\tau \in\cD_1$. Hence we can conclude that $G_1$ and $G_0$ admit the same factorization decomposition. The above argument also works for $G_2$ and $G_0$. Thus, $G_1$ and $G_2$ admit the same factorization decomposition.
\end{proof}

We now define a graphical procedure (call it \textit{redirection procedure}) that is consistent with the intervention formula in Equation \eqref{intervention-CG} and \eqref{intervention-LWF}. Let $G$ be a chain graph. Given an intervened set of variables $A \subseteq V(G)$, let $\hg$ be the chain graph obtained from $G$ by performing the following operation: for every $u \in A$ and every undirected edge $e = \{u,w\}$ containing $u$, replace $e$ by a directed edge from $u$ to $w$; finally remove all the directed edges that point to some vertex in $A$. By replacing the undirected edge with a directed edge, we replace any feedback mechanisms that include a variable in $A$ with a causal mechanism.  The intuition behind the procedure is the following. Since a variable that is set by intervention cannot be modified, the symmetric feedback relation is turned into an asymmetric causal one. Similarly, we can justify this graphical procedure as equivalent to removing the variables in $A$ from some equations in the Gibbs process on top of p. 338 of \cite{LR}, as Lauritzen and Richardson \cite{R2018} did for Equation (18) in \cite{LR}.

\begin{theorem}\label{thm:intervene-CG}
Let $G$ be a chain graph with a subset of variables $A \subseteq V(G)$ set by intervention such that for every $a \in A$. $x_a = a_0$. Let $\hg$ be obtained from $G$ by the redirection procedure. Then $G$ and $\hg$ admit the same factorization decomposition.
\end{theorem}
\begin{proof}

It is not hard to see that removing from $\hg$ and $G$ all vertices in $A$ and all edges incident to $A$ results in the same chain graph. Hence by Theorem \eqref{thm: CG_factorization-equivalence}, $G$ and $\hg$ admit the same factorization decomposition.



\end{proof}
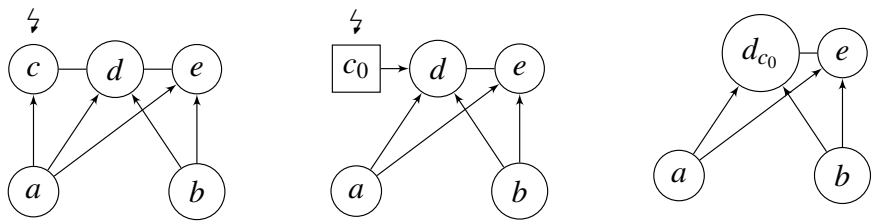
\begin{figure}[htb]
	\begin{center}
		\begin{minipage}{.2\textwidth}
		\resizebox{3cm}{!}{\input{images/interveneA.tikz}}
		\end{minipage}
		\hspace{1cm}
		\begin{minipage}{.2\textwidth}
		\resizebox{3cm}{!}{\input{images/interveneA2.tikz}}
		\end{minipage}
		\hspace{1cm}
		\begin{minipage}{.2\textwidth}
		\resizebox{3cm}{!}{\input{images/interveneA3.tikz}}
		\end{minipage}
		\caption{(a) A chain graph $G$; (b) The graph $\hg$ obtained from $G$ through the redirection procedure; (c) The graph $G_0$ obtained from $G$ by deleting variables in $A$.}
		\label{f:intervene-CG}
	\end{center} 
\end{figure}

\begin{example}
Consider the chain graph $G$ shown in Figure \ref{f:intervene-CG}. Let $\hg$ be the graph obtained from $G$ through the redirection procedure described in this subsection. Let $G_0$ be the chain graph obtained from $G$ by deleting the vertex $c_0$ and the edges incident to $c_0$. We will compare the factorization decomposition according to the formula \eqref{intervention-CG},\eqref{intervention-LWF} as well as the graph structure $\hg$ and $G_0$.

By the formula \eqref{intervention-CG} and \eqref{intervention-LWF} proposed in \cite{LR}, when $x_c$ is set as $c_0$ by intervention,
\begin{align*}
    f(x \| x_c) &= f(x_a)f(x_b)f(x_{de} \ver x_{abc_0})\\
                &= f(x_a)f(x_b) \frac{\psi_{ac_0d}(x)\psi_{abde}(x)}{\sum_{d,e} \psi_{ac_0d}(x)\psi_{abde}(x)}.
\end{align*}

Now consider the factorization according to $\hg$. The chain components of $\hg$ are $\{\{a\},\{b\},\{c\},\{d,e\}\}$ with $x_c$ set to be $c_0$. The factorization according to $\hg$ is as follows:
\begin{align*}
    f_{\hg}(x \| x_c) &= f_{\hg}(x_a) f_{\hg}(x_b) f_{\hg}(x_c) f_{\hg}(x_{de} \ver x_{abc_0})\\
                &= f_{\hg}(x_a) f_{\hg}(x_b) f_{\hg}(x_c)\frac{\psi_{ac_0d}(x)\psi_{abde}(x)}{\sum_{d,e} \psi_{ac_0d}(x)\psi_{abde}(x)},
\end{align*}
where $f(x_c) = 1$ when $x_c = c_0$ and otherwise $0$. Hence $G$ and $\hg$ admit the same factorization.

Now consider the factorization according to $G_0$. The chain components of $G_0$ are $\{\{a\},\{b\},\{d,e\}\}$. The factorization according to $G_0$ is as follows:
\begin{align*}
    f_0(x) &= f_0(x_a) f_0(x_b) f_0(x_{de} \ver x_{ab})\\
                &= f_0(x_a) f_0(x_b) \frac{\psi_{ad}(x)\psi_{abde}(x)}{\sum_{d,e} \psi_{ad}(x)\psi_{abde}(x)},
\end{align*}
Observe that $f_0(x)$ has the same form of decomposition as $f(x \| x_c)$ since $x_c$ is set to be $c_0$ in $\psi_{ac_0d}(x)$ (with the understanding that the probability of any configuration of variables with $x_c\neq c_0$ is zero). Hence we can conclude that $G, \hg$ (with $x_c$ intervened) and $G_0$ admit the same factorization decomposition.

\end{example}

\subsection{Factorization equivalence and intervention in Bayesian hypergraphs}

Intervention in Bayesian hypergraphs can be modeled analogously to the case of chain graphs. We use the same notation as before. Let $\cH$ be a DAH and $\{\tau: \tau \in \cD\}$ be its chain components. Moreover, assume further that a subset $A$ of variables in $V(\cH)$ are set such that for every $a\in A$, $x_a = a_0$. Then a probability density $f$ factorizes according to $\cH$ (with $A$ intervened) as follows: (where it is understood that the probability of any configuration of variables inconsistent with the intervention is zero): 
\begin{equation}\label{bh:intervention1}
f(x \| x_A) = \ds\prod_{\tau \in \cD} f(x_{\tau \backslash A} \ver x_{pa(\tau)}, x_{\tau \cap A}).
\end{equation}

For each $\tau \in \cD$, define $\cH_{\tau}^*$ to be the subhypergraph of $\cH_{\tau \cup pa_{\cD}(\tau)}$ containing all edges $h$ in $\cH_{\tau \cup pa(\tau)}$ such that $H(h) \subseteq \tau$, then

\begin{equation}\label{bh:intervention2}
  f(x_{\tau \backslash A} \ver x_{pa(\tau)}, x_{\tau \cap A}) = Z^{-1}(x_{pa(\tau)}, x_{\tau \cap A}) \ds\prod_{h \in \cM(\cH_{\tau}^*)} \psi_h(x).
\end{equation}
where $Z^{-1}(x_{pa(\tau)},x_{\tau \cap A}) = \ds\int_{\cX_{\tau \backslash A}} \prod_{h \in \cM(\cH_{\tau}^*)} \psi_h(x) \mu_{\tau \backslash A} (dx_{\tau \backslash A})$ and $\psi_h$ are non-negative functions that depend only on $x_h$.

\begin{definition}
Let $\cH_1$ and $\cH_2$ be two Bayesian hypergraphs.
Given a subset of variables $A_1 \subseteq V(\cH_1)$ and $A_2 \subseteq V(\cH_2)$,
we say $(\cH_1, A_1)$ and $(\cH_2, A_2)$ are \emph{factorization-equivalent} if performing the following operations to $\cH_1$ and $\cH_2$ results in the same directed acyclic hypergraph:
\begin{enumerate}[(i)]\label{BH-fact-equiv}
    \item Deleting all hyperedges with empty head, i.e., hyperedges of the form $(S,\emptyset)$.
    \item Deleting every hyperedge that is contained in some other hyperedge, i.e., delete $h$ if there is another $h'$ such that $T(h)\subseteq T(h')$ and $H(h)\subseteq H(h')$.
    \item \textit{Shrinking} all hyperedges of $\cH_i$ containing vertices in $A_i$, i.e. replace every hyperedge $h$ of $\cH_i$ by $h'=(T(h)\backslash A_i, H(h)\backslash A_i)$ for $i \in  \{1,2\}$.
\end{enumerate}
Typically, $A$ is a set of constant variables in $V$ created by intervention.
\end{definition}

\begin{theorem}\label{thm: BH_factorization-equivalence}
Let $\cH_1$ and $\cH_2$ be two DAHs defined on the same set of variables $V$. Moreover, a common set of variables $A$ in $V$ are set by intervention such that for every $a \in A$, $X_a = a_0$. If $(\cH_1,A)$ and $(\cH_2,A)$ are factorization-equivalent, then $\cH_1$ and $\cH_2$ admit the same factorization decomposition.
\end{theorem}

\begin{proof}
Similar to the proof in Theorem \ref{thm: CG_factorization-equivalence}, let $\cH_0$ be the DAH obtained from $\cH_1$ (or $\cH_2$) by performing the operations above repeatedly.
Let $\cD_1$ and $\cD_0$ be the set of chain components of $\cH_1$ and $\cH_0$ respectively.
First, note that performing the operation $(i)$ does not affect the factorization since hyperedges of the form $h = (S,\emptyset)$ never appear in the factorization decomposition due to the fact that $H(h) \cap \tau = \emptyset$ for every $\tau \in \cD_1$.
Secondly, $(ii)$ does not change the factorization decomposition too since if one hyperedge $h$ is contained in another hyperedge $h'$ as defined, then $\psi_h(x)$ can be simply absorbed into $\psi_{h'}(x)$ by replacing $\psi_{h'}(x)$ with $\psi_{h'}(x) \cdot \psi_{h}(x)$.

Now let $\tau \in \cD_1$ be an arbitrary chain component of $\cH_1$ and $h_1 \in \cH_1[\tau]^*$, i.e., the set of hyperedges in $\cH_1$ whose head intersects $\tau$. Suppose that $\tau$ is separated into several chain components $\tau'_1, \tau'_2, \cdots, \tau'_t$ in $\cH_0$ because of the shrinking operation. If $h_1\cap A = \emptyset$, then $h_1$ is also a hyperedge in $\cH_0[\tau \backslash A]^*$. If $h_1 \cap A \neq \emptyset$, there are two cases:
\begin{description}
    \item Case 1: $H(h_1) \subseteq A$. Then since variables in $A$ are constants, it follows that in Equation \eqref{bh:intervention2}, $\psi_{h_1}(x)$ does not depend on variables in $\tau\backslash A$. Hence $\psi_h(x)$ appears as factors independent of variables in $\tau\backslash A$ in both $Z^{-1}(x_{pa(\tau)},x_{\tau \cap A})$ and $\ds\prod_{h \in \cM(\cH_{\tau}^*)} \psi_h(x)$, thus cancels out with itself. 
    Note that, $h_1$ does not exist in $\cH_0$ too since $h_1$ becomes a hyperedge with empty head after being shrinked and thus is deleted in Operation (i).
    
    \item Case 2: $H(h_1)\backslash A \neq \emptyset$. In this case, $H(h_1) \backslash A$ must be entirely contained in one of $\{\tau'_1, \cdots, \tau'_t\}$ . Without loss of generality, say $H(h_1)\backslash A  \subseteq \tau'_1$ in $\cH_0$. Then note that $h_1 \backslash A$ must be contained in some maximal hyperedge $h'$ in $E(\cH_0)$ such that $H(h') \cap \tau'_1 \neq \emptyset$. Moreover, recall that variables in $A$ are constants. Hence $\psi_{h_1}$ must appear in some factor in the factorization of $f$ according to $\cH_0$.
\end{description}

Thus it follows that  every probability density $f$ that factorizes according to $\cH_1$ also factorizes according to $\cH_0$.
 On the other hand, it is not hard to see that for every $\tau' \in \cD_0$ and every hyperedge $h'$ in $(\cH_0[\tau'])^*$, $h'$ is contained in some maximal hyperedge $h$ in $(\cH_1[\tau])^*$ for some $\tau \in\cD_1$. Hence we can conclude that $\cH_1$ and $\cH_0$ admit the same factorization decomposition. The above argument also works for $\cH_2$ and $\cH_0$. Thus, $\cH_1$ and $\cH_2$ admit the same factorization decomposition.

\end{proof}

We now present a graphical procedure (call it \textit{redirection procedure}) for modeling intervention in Bayesian hypergraph. Let $\cH$ be a DAH and $\{\tau: \tau \in \cD\}$ be its chain components. Suppose a set of variables $x_{A}$ is set by intervention. We then modify $\cH$ as follows: for each hyperedge $h \in E(\cH)$ such as $S = H(h)\cap A \neq \emptyset$, replace the hyperedge $h$ by $h' = (T(h) \cup S, H(h)\backslash S)$. If a hyperedge has empty set as its head, delete that hyperedge. 
Call the resulting hypergraph $\hh_A$. We will show that the factorization according to $\hh_A$ is consistent with Equation \eqref{bh:intervention2}. 

\begin{theorem}\label{thm:intervene-BH}
Let $\cH$ be a Bayesian hypergraph and $\{\tau: \tau \in \cD\}$ be its chain components. Given an intervened set of variables $x_A$, let $\hh_A$ be the DAH obtained from $\cH$ by replacing each hyperedge $h \in E(\cH)$ satisfying $S = H(h)\cap A \neq \emptyset$ by the hyperedge $h' = (T(h) \cup S, H(h)\backslash S)$ and removing hyperedges with empty head. Then $\cH$ and $\hh$ admit the same factorization decomposition.
\end{theorem}

\begin{proof}
This is a corollary of Theorem \eqref{thm: BH_factorization-equivalence} since performing the operations (i)(ii)(iii) in the definition of factorization-equivalence of DAH to $\cH$ and $\hh$ results in the same DAH. 
\end{proof}

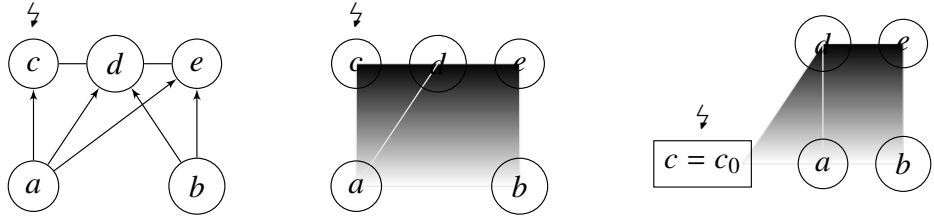
\begin{figure}[htb]
	\begin{center}
		\begin{minipage}{.2\textwidth}
		\resizebox{3cm}{!}{\input{images/interveneA.tikz}}
		\end{minipage}
		\hspace{1cm}
		\begin{minipage}{.2\textwidth}
		\resizebox{3cm}{!}{\input{images/interveneB.tikz}}
		\end{minipage}
		\hspace{1cm}
		\begin{minipage}{.2\textwidth}
		\resizebox{3.8cm}{!}{\input{images/interveneC.tikz}}
		\end{minipage}
		\caption{(a) A chain graph $G$; (b) the canonical LWF DAH $\cH$ of $G$; (c) the resulting hypergraph $\hh$ after performing the graphical procedure on $\cH$ when the variable $c$ is intervened.}
		\label{f:intervene}
	\end{center} 
\end{figure}
\begin{example}
Let $G$ be a chain graph as shown in Figure \ref{f:intervene}(a) and $\cH$ be the canonical LWF Bayesian hypergraph  of $G$ as shown in Figure \ref{f:intervene}(b), constructed based on the procedure in Section \ref{sec:hypergraph-extension}. $\cH$ has two directed hyperedges $(\{a\},\{c,d\})$ and $(\{a,b\},\{d,e\})$. Applying the redirection procedure for intervention in  Bayesian hypergraphs leads to the Bayesian hypergraph $\hh$ in Figure \ref{f:intervene}(c). We show that using equations \eqref{intervention-CG} and \eqref{intervention-LWF} for Figure \ref{f:intervene}(a) leads to the same result as if one uses the factorization formula for the Bayesian hypergraph in Figure \ref{f:intervene}(c). 

First, we compute $f(x||x_c)$ for chain graph in Figure \ref{f:intervene}(a). Based on equation \eqref{intervention-CG} we have:
$$f(x \| x_c) = f(x_a)f(x_b)f(x_{de} \ver x_{abc_0}),$$
as the effect of the atomic intervention $do(X_c=c_0)$. Then, using equation \eqref{intervention-LWF} gives:
\begin{equation}\label{ex:intervention1}
    f(x||x_c)=f(x_a)f(x_b) \frac{\psi_{ac_0d}(x)\psi_{abde}(x)}{\sum_{d,e} \psi_{ac_0d}(x)\psi_{abde}(x)}.
\end{equation}
Now, we compute $f(x)$ for Bayesian hypergraph in Figure \ref{f:intervene}(c). Using equation (\ref{eq:fact1}) gives:
$$f(x \| x_c) = f(x_a)f(x_b)f(x_{de} \ver x_{abc_0}).$$
Applying formula (\ref{eq:fact}) gives:
\begin{equation}\label{ex:intervention2}
    f(x||x_c)=f(x_a)f(x_b) f(x_c) \frac{\psi_{ac_0d}(x)\psi_{abde}(x)}{\sum_{d,e} \psi_{ac_0d}(x)\psi_{abde}(x)}
\end{equation}
Note that $f(x_c)=1$, when $x_c=c_0$, otherwise $f(x_c)=0$. As a result, the right side of equations (\ref{ex:intervention1}) and (\ref{ex:intervention2}) are the same.
\end{example}

\begin{figure}[htb]
	\begin{center}
    \includegraphics[width=.80\textwidth]{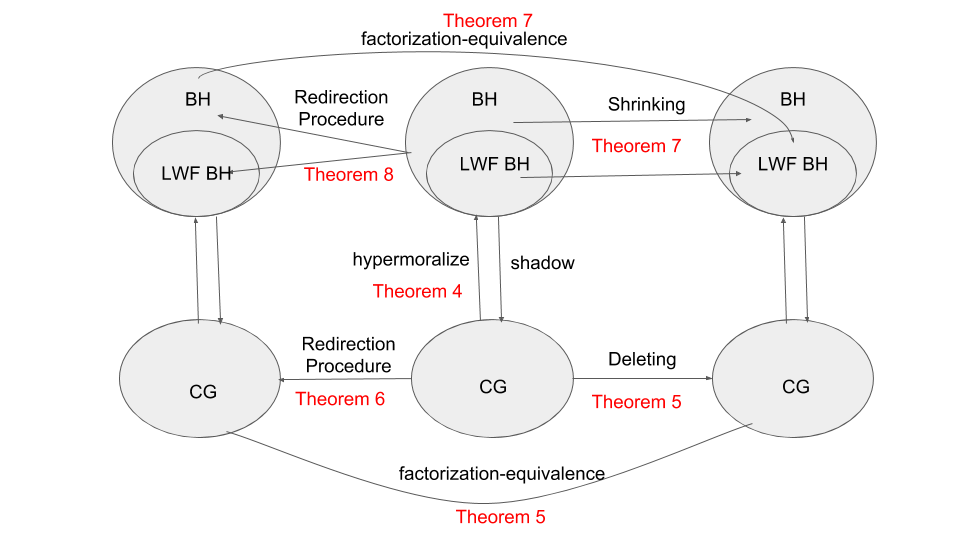}
    \caption{Commutative diagram of factorization equivalence}
		\label{f:intervention-equiv}
    \end{center}
\end{figure}

\begin{remark}
Figure \ref{f:intervention-equiv} summarizes all the results in Section \ref{sec:intervention}.
Given a chain graph $G$ and its canonical LWF DAH $\cH$, Theorem \ref{thm:factor-equiv} shows that $G$ and $\cH$ admit the same factorization decomposition. Suppose a set of variables $A$ is set by intervention. Theorem \ref{thm: CG_factorization-equivalence} and \ref{thm:intervene-CG} show that the the DAH obtained from $G$ by the redirection procedure or deleting the variables in $A$ admit the same factorization decomposition, which is also consistent with the intervention formula introduced in \cite{LR}. Similarly, Theorem \ref{thm: BH_factorization-equivalence} and \ref{thm:intervene-BH} show that the DAH obtained from $\cH$ by the redirection procedure or shrinking the variables in $A$ admit the same factorization decomposition, which is consistent with a hypergraph analogue of the formula in \cite{LR}.

\end{remark}

\section{Conclusion and Future Work}

This paper presents Bayesian hypergraph, a new probabilistic graphical model.  We showed that the model generalizes Bayesian networks, Markov networks, and LWF chain graphs, in the following sense: when the shadow of a Bayesian hypergraph is a chain graph, its Markov properties are the same as that of its shadow. We extended the causal interpretation of LWF chain graphs to Bayesian hypergraphs and provided corresponding formulas and two graphical procedures for intervention (as defined in \cite{LR}).

Directed acyclic hypergraphs can admit much finer factorizations than chain graphs, thus are more computationally efficient. The Bayesian hypergraph model also allows simpler and more general procedures for factorization as well as intervention. Furthermore, it allows a modeler to express independence of causal influence and other useful patterns, such as Noisy-OR, directly (i.e., graphically), rather than through annotations or the structure of a conditional probability table or function. We conjecture that the greater expressive power of Bayesian hypergraphs can be used to represent other PGMs and plan to explore the conjecture in future work.  

Learning the structure and the parameters of Bayesian hypergraphs is another direction for future work. For this purpose, we will need to provide a criterion for Markov equivalence of Bayesian hypergraphs.  The success of constraint-based structure learning algorithms for chain graphs leads us to hope that similar techniques would work for learning Bayesian hypergraphs. Of course, one should also explore whether a closed-form decomposable likelihood function can be derived in the discrete finite case.

\section{Acknowledgements}
This work is primarily supported by Office on Naval Research grant ONR N00014-17-1-2842.


\end{document}

%% file: images/HG_ex1.tikz
\begin{tikzpicture}

\tikzset{vertex/.style = {shape=circle,draw,minimum size=1.5em}}
\tikzset{edge/.style = {->,> = latex'}}


\shadedraw[left color=white, right color=black, opacity= 5, draw=black!10!white] 
(0,0)--(0,1)--(1,0)--cycle;

\shadedraw[left color=white, right color=black, opacity = 5, draw=black!10!white] 
(0,0)--(1,0)--(1,1)--cycle;

\shadedraw[left color=white, right color=black, opacity = 5, draw=black!10!white] 
(1,1)--(2,0)--(2,1)--cycle;

\node[vertex] (a) at (0,0){$a$};
\node[vertex] (b) at (0,1){$b$};
\node[vertex] (c) at (1,0){$c$};
\node[vertex] (d) at (1,1){$d$};
\node[vertex] (e) at (2,0){$e$}; 
\node[vertex] (f) at (2,1){$f$}; 

\draw[edge] (c) to (e);

\end{tikzpicture}   

%% file: images/HG_ex1_DAG.tikz
\begin{tikzpicture}

\tikzset{vertex/.style = {shape=circle,draw,minimum size=1.5em}}
\tikzset{edge/.style = {->,> = latex'}}

\node[vertex] (a) at (0,0){$a$};
\node[vertex] (b) at (0,2){$b$};
\node[vertex] (c) at (1,1){$c,d$};
\node[vertex] (e) at (3,1){$e,f$};

\draw[edge] (a) to (c);
\draw[edge] (b) to (c);
\draw[edge] (c) to (e);

\end{tikzpicture}   

%% file: images/HG_ex1_shadow.tikz
\begin{tikzpicture}

\tikzset{vertex/.style = {shape=circle,draw,minimum size=1.5em}}
\tikzset{edge/.style = {->,> = latex'}}

\node[vertex] (a) at (0,0){$a$};
\node[vertex] (b) at (0,1){$b$};
\node[vertex] (c) at (1,0){$c$};
\node[vertex] (d) at (1,1){$d$};
\node[vertex] (e) at (2,0){$e$}; 
\node[vertex] (f) at (2,1){$f$};

\draw (d) to (c);
\draw (e) to (f);

\draw[edge] (a) to (c);
\draw[edge] (a) to (d);
\draw[edge] (b) to (c);
\draw[edge] (d) to (f);
\draw[edge] (d) to (e);
\draw[edge] (c) to (e);

\end{tikzpicture}   

%% file: images/CG_ex2.tikz
\begin{tikzpicture}
\tikzset{vertex/.style = {shape=circle,draw,minimum size=1.5em}}
\tikzset{edge/.style = {->,> = latex'}}


\node[vertex] (a) at (0,0)  {$a$};
\node[vertex] (b) at (1,-0.5)  {$b$}; 
\node[vertex] (c) at (2,0)  {$c$};
\node[vertex] (d) at (0,1.5){$d$};
\node[vertex] (e) at (1,1.5){$e$};
\node[vertex] (f) at (2,1.5){$f$};


\draw[edge] (a) to (d);
\draw[edge] (a) to (e);
\draw[edge] (b) to (e);
\draw[edge] (b) to (f);
\draw[edge] (c) to (f);

\draw (d) to (e);
\draw (e) to (f);

\end{tikzpicture}

%% file: images/HBN_ex2_moral.tikz
\begin{tikzpicture}
\tikzset{vertex/.style = {shape=circle,draw,minimum size=1.5em}}
\tikzset{edge/.style = {->,> = latex'}}


\shadedraw[bottom color=white, top color=black, opacity=5, draw=black!10!white] (0,0)--(0,1.5)--(1,1.5)--cycle;

\shadedraw[bottom color=white, top color=black, opacity=5, draw=black!10!white]  (1,-0.5)--(1,1.5)--(2,1.5)--cycle;

\shadedraw[bottom color=white, top color=black, opacity=5, draw=black!10!white]  (1,-0.5)--(2,0)--(2,1.5)--cycle;

\shadedraw[bottom color=white, top color=black, opacity=5, draw=black!10!white]  (0,0)--(1,-0.5)--(1,1.5)--cycle;

\node[vertex] (a) at (0,0)  {$a$};
\node[vertex] (b) at (1,-0.5)  {$b$}; 
\node[vertex] (c) at (2,0)  {$c$};
\node[vertex] (d) at (0,1.5){$d$};
\node[vertex] (e) at (1,1.5){$e$};
\node[vertex] (f) at (2,1.5){$f$}; 

\end{tikzpicture}

%% file: images/BH_altfact/BH_altfact_ex2.tikz
\begin{tikzpicture}
\tikzset{vertex/.style = {shape=circle,draw,minimum size=1.5em}}
\tikzset{edge/.style = {->,> = latex'}}



\node[vertex] (a) at (0,1)  {$a$};
\node[vertex] (b) at (2,1)  {$b$};
\node[vertex] (c) at (1,0)  {$c$};


\draw[edge] (a) to (c);
\draw[edge] (b) to (c);

\end{tikzpicture}

%% file: images/CG_fact_ex2.tikz
\begin{tikzpicture}
\tikzset{vertex/.style = {shape=circle,draw,minimum size=1.5em}}
\tikzset{edge/.style = {->,> = latex'}}


\node[vertex] (a) at (0,1)  {$a$};
\node[vertex] (b) at (1,1)  {$b$};
\node[vertex] (c) at (0,0)  {$c$};
\node[vertex] (d) at (1,0)  {$d$};

\draw[edge] (a) to (c);
\draw[edge] (a) to (d);
\draw[edge] (b) to (c);
\draw[edge] (b) to (d);
\draw (c) to (d);

\end{tikzpicture}

%% file: images/BH_fact_ex2.tikz
\begin{tikzpicture}
\tikzset{vertex/.style = {shape=circle,draw,minimum size=0.5mm}}
\tikzset{edge/.style = {->,> = latex'}}


\shadedraw[top color=white, bottom color=black, draw=black!10!white] (0,1)--(1,1)--(0,0)--cycle;
\shadedraw[top color=white, bottom color=black, draw=black!10!white] (0,1)--(1,1)--(1,0)--cycle;

\node[vertex] (a) at (0,1)  {$a$};
\node[vertex] (b) at (1,1)  {$b$};
\node[vertex] (c) at (0,0)  {$c$};
\node[vertex] (d) at (1,0)  {$d$}; 

\draw (c) to (d);

\end{tikzpicture}

%% file: images/BH_altfact/BH_altfact_ex1.tikz
\begin{tikzpicture}
\tikzset{vertex/.style = {shape=circle,draw,minimum size=1.5em}}
\tikzset{edge/.style = {->,> = latex'}}

\shadedraw[bottom color=black, top color=white, opacity=5, draw=black!10!white] (0,1)--(2,1)--(1,0)--cycle;

\node[vertex] (a) at (0,1)  {$a$};
\node[vertex] (b) at (2,1)  {$b$};
\node[vertex] (c) at (1,0)  {$c$};


\end{tikzpicture}

%% file: images/BH_altfact/BH_altfact_ex1B.tikz
\begin{tikzpicture}
\tikzset{vertex/.style = {shape=circle,draw,minimum size=1.5em}}
\tikzset{edge/.style = {->,> = latex'}}


\shadedraw[bottom color=black, top color=white, opacity=5, draw=black!10!white] (0,1)--(2,1)--(1,0)--cycle;

\node[vertex] (a) at (0,1)  {$a$};
\node[vertex] (b) at (2,1)  {$b$};
\node[vertex] (c) at (1,0)  {$c$};

\draw (a) to (b);

\end{tikzpicture}

%% file: images/BH_altfact/BH_altfact_ex2B.tikz
\begin{tikzpicture}
\tikzset{vertex/.style = {shape=circle,draw,minimum size=1.5em}}
\tikzset{edge/.style = {->,> = latex'}}



\node[vertex] (a) at (0,1)  {$a$};
\node[vertex] (b) at (2,1)  {$b$};
\node[vertex] (c) at (1,0)  {$c$};

\draw (a) to (b);

\draw[edge] (a) to (c);
\draw[edge] (b) to (c);

\end{tikzpicture}

%% file: images/BH_altfact/BH_altfact_ex4.tikz
\tikzset{
    >=stealth,
    font=\scriptsize,
    possible world/.style={circle,draw,thick,align=center},
    real world/.style={double,circle,draw,thick,align=center},
    minimum size=40pt
}

\begin{tikzpicture}
\tikzset{vertex/.style = {shape=circle,draw,minimum size=1.5em}}
\tikzset{edge/.style = {->,> = latex', thick}}



\node[vertex] (a) at (0,1)  {$a$};
\node[vertex] (b) at (2,1)  {$b$};
\node[vertex] (c) at (1,0)  {$c$};


\draw[edge] (a) to (c);

\end{tikzpicture}

%% file: images/BH_altfact/BH_altfact_ex4B.tikz
\tikzset{
    >=stealth,
    font=\scriptsize,
    possible world/.style={circle,draw,thick,align=center},
    real world/.style={double,circle,draw,thick,align=center},
    minimum size=40pt
}

\begin{tikzpicture}
\tikzset{vertex/.style = {shape=circle,draw,minimum size=1.5em}}
\tikzset{edge/.style = {->,> = latex', thick}}



\node[vertex] (a) at (0,1)  {$a$};
\node[vertex] (b) at (2,1)  {$b$};
\node[vertex] (c) at (1,0)  {$c$};

\draw (a) to (b);

\draw[edge] (a) to (c);

\end{tikzpicture}

%% file: images/BH_altfact/BH_altfact_ex5.tikz
\tikzset{
    >=stealth,
    font=\scriptsize,
    possible world/.style={circle,draw,thick,align=center},
    real world/.style={double,circle,draw,thick,align=center},
    minimum size=40pt
}

\begin{tikzpicture}
\tikzset{vertex/.style = {shape=circle,draw,minimum size=1.5em}}
\tikzset{edge/.style = {->,> = latex', thick}}



\node[vertex] (a) at (0,1)  {$a$};
\node[vertex] (b) at (2,1)  {$b$};
\node[vertex] (c) at (1,0)  {$c$};


\end{tikzpicture}

%% file: images/BH_altfact/BH_altfact_ex5B.tikz
\tikzset{
    >=stealth,
    font=\scriptsize,
    possible world/.style={circle,draw,thick,align=center},
    real world/.style={double,circle,draw,thick,align=center},
    minimum size=40pt
}

\begin{tikzpicture}
\tikzset{vertex/.style = {shape=circle,draw,minimum size=1.5em}}
\tikzset{edge/.style = {->,> = latex', thick}}



\node[vertex] (a) at (0,1)  {$a$};
\node[vertex] (b) at (2,1)  {$b$};
\node[vertex] (c) at (1,0)  {$c$};

\draw (a) to (b);

\end{tikzpicture}

%% file: images/BH_altfact/BH_altfact_ex6.tikz
\begin{tikzpicture}
\tikzset{vertex/.style = {shape=circle,draw,minimum size=1.5em}}
\tikzset{edge/.style = {->,> = latex'}}



\node[vertex] (a) at (0,1)  {$a$};
\node[vertex] (b) at (2,1)  {$b$};
\node[vertex] (c) at (1,0)  {$c$};

\draw (a) to (c);
\draw (b) to (c);


\end{tikzpicture}

%% file: images/BH_altfact/BH_altfact_ex6B.tikz
\begin{tikzpicture}
\tikzset{vertex/.style = {shape=circle,draw,minimum size=1.5em}}
\tikzset{edge/.style = {->,> = latex'}}



\node[vertex] (a) at (0,1)  {$a$};
\node[vertex] (b) at (2,1)  {$b$};
\node[vertex] (c) at (1,0)  {$c$};

\draw (a) to (b);

\draw[edge] (c) to (a);

\end{tikzpicture}

%% file: images/BH_altfact/BH_altfact_ex7.tikz
\begin{tikzpicture}
\tikzset{vertex/.style = {shape=circle,draw,minimum size=1.5em}}
\tikzset{edge/.style = {->,> = latex'}}



\node[vertex] (a) at (0,1)  {$a$};
\node[vertex] (b) at (2,1)  {$b$};
\node[vertex] (c) at (1,0)  {$c$};


\draw[edge] (c) to (a);
\draw[edge] (c) to (b);

\end{tikzpicture}

%% file: images/BH_altfact/BH_altfact_ex7B.tikz
\begin{tikzpicture}
\tikzset{vertex/.style = {shape=circle,draw,minimum size=1.5em}}
\tikzset{edge/.style = {->,> = latex'}}



\node[vertex] (a) at (0,1)  {$a$};
\node[vertex] (b) at (2,1)  {$b$};
\node[vertex] (c) at (1,0)  {$c$};

\draw (a) to (b);

\draw[edge] (c) to (a);
\draw[edge] (c) to (b);

\end{tikzpicture}

%% file: images/obesity.tikz
\begin{tikzpicture}
\tikzset{vertex/.style = {shape=circle,draw,minimum size=1.5em}}
\tikzset{tuoyuan/.style= {shape=ellipse,draw, minimum width=3.5cm, minimum height=1.5cm}}
\tikzset{edge/.style = {->,> = latex',ultra thick}}

\node[tuoyuan] (a) at (0,2)  {\Large Moderate Exercise};
\node[tuoyuan] (b) at (6,2)  {\Large Diet};
\node[tuoyuan] (c) at (3,0)  {\Large Obesity};

\draw[edge] (a) to (c);
\draw[edge] (b) to (c);

\end{tikzpicture}

%% file: images/AtherA.tikz
\begin{tikzpicture}
\tikzset{vertex/.style = {shape=circle,draw,minimum size=1.5em}}
\tikzset{tuoyuan/.style= {shape=ellipse,draw, minimum width=3cm, minimum height=1.5cm}}
\tikzset{edge/.style = {->,> = latex',ultra thick}}



\node[tuoyuan] (a) at (-1,3)  {\Large Moderate Exercise};
\node[tuoyuan] (b) at (4,3)  {\Large S-LDL};
\node[tuoyuan] (c) at (8,3)  {\Large S-T};
\node[tuoyuan] (d) at (12,3)  {\Large C-HDL};
\node[tuoyuan] (e) at (6,0)  {\Large Atheriosclerosis};
\node[tuoyuan] (f) at (6,6)  {\Large Diet};

\draw[edge] (a) to (e);
\draw[edge] (b) to (e);
\draw[edge] (c) to (e);
\draw[edge] (d) to (e);
\draw[edge] (f) to (b);
\draw[edge] (f) to (c);
\draw[edge] (f) to (d);

\end{tikzpicture}

%% file: images/AtherB.tikz
\begin{tikzpicture}
\tikzset{vertex/.style = {shape=circle,draw,minimum size=1.5em}}
\tikzset{tuoyuan/.style= {shape=ellipse,draw, minimum width=3cm, minimum height=1.5cm}}
\tikzset{edge/.style = {->,> = latex',ultra thick}}

\node[tuoyuan] (a) at (-1,3)  {\Large Moderate Exercise};
\node[tuoyuan] (b) at (4,4)  {\Large S-LDL};
\node[tuoyuan] (c) at (8,3)  {\Large S-T};
\node[tuoyuan] (d) at (12,4)  {\Large C-HDL};
\node[tuoyuan] (e) at (6,0)  {\Large Atheriosclerosis};

\draw[edge] (a) to (e);
\draw[edge] (b) to (e);
\draw[edge] (c) to (e);
\draw[edge] (d) to (e);

\fill[gray, opacity=0.5] (4,4)--(8,3)--(12,4)--cycle;

\end{tikzpicture}

%% file: images/BH_fact_ex2A.tikz
\begin{tikzpicture}
\tikzset{vertex/.style = {shape=circle,draw,minimum size=0.5mm}}
\tikzset{edge/.style = {->,> = latex'}}

\shadedraw[top color=white, bottom color=black, draw=black!10!white] (0,1)--(1,1)--(1,0)--(0,0)--cycle;

\node[vertex] (a) at (0,1)  {$a$};
\node[vertex] (b) at (1,1)  {$b$};
\node[vertex] (c) at (0,0)  {$c$};
\node[vertex] (d) at (1,0)  {$d$}; 

\draw (c) to (d);

\end{tikzpicture}

%% file: images/CG_ex3.tikz
\begin{tikzpicture}
\tikzset{vertex/.style = {shape=circle,draw,minimum size=1.5em}}
\tikzset{edge/.style = {->,> = latex'}}


\node[vertex] (a) at (0,0)  {$a$};
\node[vertex] (b) at (1,-0.5)  {$b$}; 
\node[vertex] (c) at (2,0)  {$c$};
\node[vertex] (d) at (0,1.5){$d$};
\node[vertex] (e) at (1,1.5){$e$};
\node[vertex] (f) at (2,1.5){$f$};


\draw[edge] (a) to (d);
\draw[edge] (a) to (e);
\draw[edge] (b) to (e);
\draw[edge] (c) to (f);
\draw[edge] (c) to (e);

\draw (d) to (e);
\draw (e) to (f);

\end{tikzpicture}

%% file: images/CG_ex3_moral.tikz
\begin{tikzpicture}
\tikzset{vertex/.style = {shape=circle,draw,minimum size=1.5em}}
\tikzset{edge/.style = {->,> = latex'}}


\node[vertex] (a) at (0,0)  {$a$};
\node[vertex] (b) at (1,-0.5)  {$b$}; 
\node[vertex] (c) at (2,0)  {$c$};
\node[vertex] (d) at (0,1.5){$d$};
\node[vertex] (e) at (1,1.5){$e$};
\node[vertex] (f) at (2,1.5){$f$};


\draw (a) to (d);
\draw (a) to (e);
\draw (b) to (e);
\draw (c) to (e);
\draw (c) to (f);
\draw (d) to (e);
\draw (e) to (f);
\draw (a) to (b);
\draw (a) to (c);
\draw (b) to (c);

\end{tikzpicture}

%% file: images/HBN_ex3_fact.tikz
\begin{tikzpicture}
\tikzset{vertex/.style = {shape=circle,draw,minimum size=1.5em}}
\tikzset{edge/.style = {->,> = latex'}}

\shadedraw[bottom color=white, top color=black, opacity=5, draw=black!10!white] (0,0)--(0,1.5)--(1,1.5)--cycle;

\shadedraw[bottom color=white, top color=black, opacity=5, draw=black!10!white]  (0,0)--(1,-0.5)--(2,0)--(1,1.5)--cycle;

\shadedraw[bottom color=white, top color=black, opacity=5, draw=black!10!white]  (1,1.5)--(2,1.5)--(2,0)--cycle;

\node[vertex] (a) at (0,0)     {$a$};
\node[vertex] (b) at (1,-0.5)  {$b$}; 
\node[vertex] (c) at (2,0)     {$c$};
\node[vertex] (d) at (0,1.5)   {$d$};
\node[vertex] (e) at (1,1.5)   {$e$};
\node[vertex] (f) at (2,1.5)   {$f$};

\end{tikzpicture}

%% file: images/interveneA.tikz
\begin{tikzpicture}
\tikzset{vertex/.style = {shape=circle,draw,minimum size=1.5em}}
\tikzset{edge/.style = {->,> = latex'}}


\node[vertex] (a) at (0,0)  {$a$};
\node[vertex] (b) at (2,0)  {$b$};
\node[vertex,label=\lightning] (c) at (0,1.5)  {$c$};
\node[vertex] (d) at (1,1.5)  {$d$};
\node[vertex] (e) at (2,1.5)  {$e$};

\draw[edge] (a) to (c);
\draw[edge] (a) to (d);
\draw[edge] (a) to (e);
\draw[edge] (b) to (d);
\draw[edge] (b) to (e);

\draw (c) to (d);
\draw (e) to (d);

\end{tikzpicture}

%% file: images/interveneA2.tikz
\begin{tikzpicture}
\tikzset{vertex/.style = {shape=circle,draw,minimum size=1.5em}}
\tikzset{edge/.style = {->,> = latex'}}


\node[vertex] (a) at (0,0)  {$a$};
\node[vertex] (b) at (2,0)  {$b$};
\node[vertex, rectangle, label=\lightning] (c) at (0,1.5)  {$c_0$};
\node[vertex] (d) at (1,1.5)  {$d$};
\node[vertex] (e) at (2,1.5)  {$e$};

\draw[edge] (c) to (d);
\draw[edge] (a) to (d);
\draw[edge] (a) to (e);
\draw[edge] (b) to (d);
\draw[edge] (b) to (e);

\draw (e) to (d);

\end{tikzpicture}

%% file: images/interveneA3.tikz
\begin{tikzpicture}
\tikzset{vertex/.style = {shape=circle,draw,minimum size=1.5em}}
\tikzset{edge/.style = {->,> = latex'}}


\node[vertex] (a) at (0,0)  {$a$};
\node[vertex] (b) at (2,0)  {$b$};
\node[vertex] (d) at (1,1.5)  {$d_{c_0}$};
\node[vertex] (e) at (2,1.5)  {$e$};

\draw[edge] (a) to (d);
\draw[edge] (a) to (e);
\draw[edge] (b) to (d);
\draw[edge] (b) to (e);

\draw (e) to (d);

\end{tikzpicture}

%% file: images/interveneB.tikz
\begin{tikzpicture}
\tikzset{vertex/.style = {shape=circle,draw,minimum size=1.5em}}
\tikzset{edge/.style = {->,> = latex'}}

\shadedraw[bottom color=white, top color=black, opacity=5, draw=black!10!white] (0,0)--(0,1.5)--(1,1.5)--cycle;

\shadedraw[bottom color=white, top color=black, opacity=5, draw=black!10!white] (0,0)--(2,0)--(2,1.5)--(1,1.5)--cycle;

\node[vertex] (a) at (0,0)  {$a$};
\node[vertex] (b) at (2,0)  {$b$};
\node[vertex,label=\lightning] (c) at (0,1.5)  {$c$};
\node[vertex] (d) at (1,1.5)  {$d$};
\node[vertex] (e) at (2,1.5)  {$e$};

\end{tikzpicture}

%% file: images/interveneC.tikz
\begin{tikzpicture}
\tikzset{vertex/.style = {shape=circle,draw,minimum size=1.5em}}
\tikzset{edge/.style = {->,> = latex'}}

\shadedraw[bottom color=white, top color=black, opacity=5, draw=black!10!white] (1,0)--(-0,0)--(1,1.5)--cycle;

\shadedraw[bottom color=white, top color=black, opacity=5, draw=black!10!white] (1,0)--(2,0)--(2,1.5)--(1,1.5)--cycle;

\node[vertex] (a) at (1,0)  {$a$};
\node[vertex] (b) at (2,0)  {$b$};
\node[vertex,rectangle, label=\lightning] (c) at (-0.5,0)  {$c= c_0$};
\node[vertex] (d) at (1,1.5)  {$d$};
\node[vertex] (e) at (2,1.5)  {$e$};

\end{tikzpicture}

%% file: main.bbl
\begin{thebibliography}{1}

\bibitem{AMP}
Andersson, S. A., Madigan, D. and Perlman, M. D. (2001). Alternative Markov properties for chain graphs. \textit{Scand. J. Stat.} \textbf{28} 33-85.

\bibitem{CW}
Cox, D. R. and Wermuth, N. (1993). Linear dependencies represented by chain graphs (with discussion). \textit{Statist. Sci..} \textbf{8} 204-218; 247-277.

\bibitem{CWm}
Cox, D. R. and Wermuth, N. (1996). Multivariate Dependencies. London: Chapman \& Hall.

\bibitem{DLS}
Darroch, J. N., Lauritzen, S. L. and Speed, T. P. (1980). Markov fields and log-linear interaction models for contingency tables. \textit{Ann. Statist.} \textbf{8} 522-539.

\bibitem{Dawid}
 Dawid, A. P. (1980). Conditional independence for statistical operations. \textit{Ann. Statist}. \textbf{8}, 598-617.

\bibitem{Drton}
Drton, M. (2009). Discrete chain graph models. \textit{Bernoulli} \textbf{15} 736-753.

\bibitem{Fry}
Frydenberg, M. (1990). The chain graph Markov property. \textit{Scand. J. Statist.} \textbf{17} 333-353.

\bibitem{GV00}
Ghosh, J.K. and M. Valtorta.
Building a Bayesian Network Model of Heart Disease.
\textit{Proceedings of the 38th Annual ACM Southeastern Conference}, Clemson, South Carolina, 239-240, 2000. (Extended version at https://cse.sc.edu/~mgv/reports/tr1999-11.pdf.)

\bibitem{HC}
Hammersley, J.M. and Clifford, P.E. (1971). Markov fields on finite graphs and lattices. Unpublished manuscript.

\bibitem{Isham}
Isham, V. (1981). An introduction to spatial point processes and Markov random fields. \textit{Internat. Statist. Rev}. \textbf{49}, 21-43.


\bibitem{Hajek92}
Petr Hajek, Tomas Havranek, and Radim Jirousek (1992) \textit{Uncertain Information Processing In Expert Systems.} Boca Raton, FL: CRC Press.

\bibitem{JN07}
Jensen, F.V. and Nielsen, T.D. (2007) \textit{Bayesian Networks and Decision Graphs, 2nd ed.} New York: Springer.

\bibitem{KSC}
Kiiveri, H., Speed, T. P. and Carlin, J. B. (1984). Recursive causal models. \textit{J. Aust. Math. Soc. Ser. A} \textbf{36} 30-52.

\bibitem{Koster}
Koster, J.T.A. (1999). On the validity of the Markov interpretation of path diagrams of Gaussian structural equation systems with correlated errors. \textit{Scand. J. Statist.} \textbf{26} 413-431.

\bibitem{LS}
Lauritzen, S. L. \& Spiegelhalter, D. J. (1988). Local computations with probabilities on graphical
structures and their application to expert systems. \textit{J. Roy. Statist. Soc. Ser. B} \textbf{50}, 157-224.

\bibitem{Letal}
Lauritzen S.L., Dawid A.P., Larsen B.N., and Leimer H.G. 1990. Independence properties of directed Markov fields. \textit{Networks} \textbf{20},
491-505.

\bibitem{Lau}
Lauritzen, S. L., \textit{Graphical Models}, New York: Clarendon, 1996.

\bibitem{LJ}
Lauritzen, S.L. and Jensen, F.V. (1997)
Local Computations with Valuations from a Commutative Semigroup.
\textit{Ann. Math. Artificial Intelligence}, \textbf{21}, 51-69.

\bibitem{LR}
Lauritzen, S.L. and Richardson, T.S. (2002). Chain graph models and their causal
 interpretations (with discussion). \textit{J. Roy. Statist. Soc. Ser B}, \textbf{64}, 321-361.

\bibitem{LW}
Lauritzen, S. L. and Wermuth, N. (1989). Graphical models for association between variables, some of which are qualitative and some quantitative. \textit{Ann. Statist.} \textbf{17} 31-57.

\bibitem{Neyman}
Neyman, J. (1923) \textit{On the Application of Probability Theory to Agricultural Experiments: Essay on Principles}.
 (in Polish) (Engl. transl. D. Dabrowska and T. P. Speed, \textit{Statist. Sci.}, \textbf{5} (1990), 465-480.


\bibitem{Pearl86}
 Pearl, J. and Paz, A. (1986). Graphoids. A graph-based logic for reasoning about relevancy relations.
 \textit{Proceedings of 7th European Conference on Artificial Intelligence}, Brighton, United Kingdom, June 1986.

\bibitem{Pearl88}
Pearl, J. (1988) \textit{Probabilistic Reasoning in Intelligent Systems.} San Mateo, CA: Morgan-Kaufmann.

\bibitem{Pearl93}
Pearl, J. (1993) Graphical models, causality and intervention. \textit{Statist. Sci.}, \textbf{8}, 266-269.

\bibitem{Pearl95}
Pearl, J.  (1995) Causal diagrams for empirical research. \textit{Biometrika}, \textbf{82}, 669-710. 

\bibitem{Pearl20}
Pearl, J.  (2009) \textit{Causality: Models, Reasoning, and Inference, 2nd ed.} Cambridge: Cambridge University Press.


\bibitem{Pena}
Pena, J. M. (2014). Marginal AMP chain graphs. \textit{Internat. J. Approx. Reason.} \textbf{55}, 1185-1206.

\bibitem{Rich}
Richardson, T. S. (2003). Markov properties for acyclic directed mixed graphs. \textit{Scand. J. Statist.} \textbf{30} 145-157.

\bibitem{R2018}
Richardson, T. S. (2018). \textit{Personal communication}.

\bibitem{RS}
Richardson, T. S. and Spirtes, P. (2002). Ancestral graph Markov models. \textit{Ann. Statist.} \textbf{30} 962-1030.

\bibitem{Rudin}
Rubin, D. B. (1974) Estimating causal effects of treatments in randomized and non-randomized studies. \textit{J. Educ. Psychol.}, \textbf{66}, 688-701.


\bibitem{Speed}
Speed, T. P. (1979). A note on nearest-neighhbour Gibbs and Markov probabilities. \textit{Sankhya A} \textbf{41}, 184-197.

\bibitem{Sperner}
E. Sperner, Ein Satz \"{u}ber Untermengen einer endlichen Menge. \textit{Math. Z.} \textbf{27} (1928), 544-548.

\bibitem{Spirtes}
Spirtes, P., Glymour, C. and Scheines, R. (1993) \textit{Causation, Prediction and Search}. New York: Springer.


\bibitem{Studeny92}
Studen\'y, M. (1992). Conditional independence relations have no finite complete  characterization. in S. Kub\'ik and J.\'A. V\'i\v sek (eds.), \textit{Information Theory, Statistical Decision Functions and Random Processes: Proceedings of the 11th Prague Conference - B}, Kluwer, Dordrecht (also Academia, Prague),  \textbf{15} 377-396.

\bibitem{Studeny09}
Studen\'y, Milan; Roverato, Alberto; and \v{S}t\v{e}p\'anov\'a, \v{S}\'arka. Two operations of merging and splitting components in a chain graph. Kybernetika (Prague) 45 (2009), no. 2, 208--248

\bibitem{WC}
Wermuth, N. and Cox, D.R. (2004). Joint response graphs and separation induced by triangular systems. \textit{J. R. Stat. Soc. Ser. B Stat. Methodol.} \textbf{66} 687-717.

\bibitem{WCP}
Wermuth, N., Cox, D. R. and Pearl, J. (1994). Explanation for multivariate structures derived from univariate recursive regressions Technical Report No. 94(1), Univ. Mainz, Germany.

\bibitem{WL}
Wermuth, N. and Lauritzen, S. L. (1983). Graphical and recursive models for contingency tables. \textit{Biometrika} \textbf{70} 537-552.

\end{thebibliography}
